\documentclass[a4paper,onecolumn,11pt]{quantumarticle}
\pdfoutput=1
\usepackage[utf8]{inputenc}
\usepackage[english]{babel}
\usepackage[T1]{fontenc}
\usepackage{amsmath}
\usepackage{amsthm}
\usepackage{hyperref}
\usepackage{caption}
\usepackage{subcaption}
\usepackage{tikz}
\usepackage{physics}
\usepackage{lipsum}
\usepackage{bm}
\usepackage{stmaryrd}
\usepackage{amssymb}
\usepackage{xcolor}
\usepackage{soul}
\usepackage{ulem}
\usepackage{fontawesome}
\usepackage{eczoo}
\PassOptionsToPackage{compress}{natbib}
\usepackage[numbers]{natbib}
\theoremstyle{definition}
\newtheorem{proposition}{Proposition}
\newtheorem{corollary}{Corollary}
\newtheorem{definition}{Definition}
\newtheorem{lemma}{Lemma}
\newtheorem{theorem}{Theorem}
\newtheorem{eg}{Example}

\begin{document}

\title{On stability of k-local quantum phases of matter}

\author{Ali Lavasani}
\affiliation{Joint Quantum Institute, NIST/University of Maryland, College Park, MD 20742, USA}
\affiliation{Condensed Matter Theory Center, University of Maryland, College Park, MD 20742, USA}
\affiliation{Kavli Institute for Theoretical Physics, University of California, Santa Barbara, CA 93106, USA}

\author{Michael J. Gullans}
\affiliation{Joint Center for Quantum Information and Computer Science,
NIST/University of Maryland, College Park, MD 20742, USA}

\author{Victor V. Albert}
\affiliation{Joint Center for Quantum Information and Computer Science,
NIST/University of Maryland, College Park, MD 20742, USA}

\author{Maissam Barkeshli}
\affiliation{Joint Quantum Institute, NIST/University of Maryland, College Park, MD 20742, USA}
\affiliation{Condensed Matter Theory Center, University of Maryland, College Park, MD 20742, USA}
\affiliation{Department of Physics, University of Maryland, College Park, MD 20742, USA }

\maketitle

\begin{abstract}
The current theoretical framework for topological phases of matter is based on the thermodynamic limit of a system with geometrically local interactions. 
A natural question is to what extent the notion of a phase of matter remains well-defined if we relax the constraint of geometric locality, and replace it with a weaker graph-theoretic notion of $k$-locality. 
As a step towards answering this question, we analyze the stability of the energy gap to perturbations for Hamiltonians corresponding to general quantum low-density parity-check codes, extending work of Bravyi and Hastings~\href{https://doi.org/10.1007/s00220-011-1346-2}{[Commun.\@ Math.\@ Phys.\@ \textbf{307}, 609 (2011)]}.
A corollary of our main result is that if there exist constants $\varepsilon_1,\varepsilon_2>0$ such that the size $\Gamma(r)$ of balls of radius $r$ on the interaction graph satisfy $\Gamma(r) = O(\exp(r^{1-\varepsilon_1}))$ and the local ground states of balls of radius $r\le\rho^\ast = O(\log(n)^{1+\varepsilon_2})$ are locally indistinguishable, then the energy gap of the associated Hamiltonian is stable against local perturbations. 
This gives an almost exponential improvement over the \(D\)-dimensional Euclidean case, which requires $\Gamma(r) = O(r^D)$ and $\rho^\ast = O(n^\alpha)$ for some \(\alpha\).
The approach we follow falls just short of proving stability of finite-rate qLDPC codes, which have $\varepsilon_1 = 0$; we discuss some strategies to extend the result to these cases. 
We discuss implications for the third law of thermodynamics, as $k$-local Hamiltonians can have extensive zero-temperature entropy.  
\end{abstract}

\section{Introduction}

The last several decades in physics have yielded significant advances in our understanding of quantum phases of matter, in particular topological phases of matter \cite{wen04,nayak2008,wang2008,hasan2010,qi2011,bernevig2013topological,senthil2015,zeng2019quantum,barkeshli2019,sachdev2023quantum}. A fundamental aspect of the theory is an underlying assumption of geometric locality. That is, it is assumed that the microscopic degrees of freedom exist on a well-behaved discretization or in the continuum of a fixed Riemannian manifold, and the Hamiltonian is a sum of geometrically local operators. Phases of matter are then defined in the thermodynamic limit of infinite volume. Understanding how to characterize distinct phases is an ongoing research direction.

A natural question is to what extent we can generalize our theories by relaxing the requirement of geometric locality. Clearly we cannot do away completely with locality, as the problem then reduces to quantum mechanics in $0$ spatial dimensions, and the notion of a phase of matter breaks down. Instead we can consider relaxing geometric locality while still keeping a graph-theoretic notion of locality. For example, suppose that the Hilbert space breaks up into a tensor product of local Hilbert spaces, $\mathcal{H} = \otimes_i \mathcal{H}_i$, and the Hamiltonian is a sum of terms $H = \sum_i h_i$, where each $h_i$ has support on at most $k$ sites, irrespective of any notion of geometry.\footnote{We may also choose to require that each site appears in at most a constant $k'$ number of terms} We refer to such Hamiltonians as $k$-local. Can we now develop a theory of $k$-local phases of matter?
How does it relate to the more familiar case where we require geometric locality? For a given finite-size system, we can embed the interaction graph in a manifold such that the $k$-locality translates to geometric locality. Crucially, what makes the problem of $k$-local phases of matter distinct is that the thermodynamic limit does not correspond to an infinite volume limit of a fixed manifold with well-behaved geometry.  

There are two motivations for this question, one theoretical and one experimental. On the experimental side, there have been an increasing number of platforms, such as ultracold atoms in optical cavities, circuit QED, and ion traps, which can realize designer Hamiltonians with long-range interactions \cite{kollar2019hyperbolic,periwal2021programmable,bluvstein2022quantum,iqbal2023topological,iqbal2024non,xu2024constant,bluvstein2024logical}. Arbitrarily long-range interactions allow us to relax the geometric locality requirement and to consider more general $k$-local Hamiltonians. 

On the theoretical side, it is well-known that topological phases of matter are deeply related to \eczoohref[quantum error-correcting codes]{topological} \cite{kitaev2003fault}. Families of quantum codes can be mapped to families of Hamiltonians. In particular it is possible to have families of finite-rate codes, where the number of logical qubits scales linearly with the number of physical qubits $n$. Moreover, in recent years there have been important breakthroughs such as the discovery of \eczoohref[asymptotically good quantum low-density parity-check (qLDPC) codes]{good_qldpc}, which not only have finite rate, but also code distance $d$ scaling linearly with $n$ \cite{panteleev2022asymptotically,dinur2023good,lin2022good,leverrier2022quantum}. A natural question to ask is whether and when distinct families of quantum codes correspond to distinct $k$-local phases of matter. We note that recently, Refs. \cite{rakovszky2023physics,rakovszky2024physics} also raised a similar question in viewing general LDPC codes from a physics perspective.

Finite rate quantum codes by definition have corresponding Hamiltonians whose ground state degeneracy increases exponentially with the system size $n$. Thus there is an \it extensive \rm zero-temperature entropy $S$. Can such an extensive zero-temperature entropy be robust to perturbations?\footnote{The SYK model \cite{sachdev1993gapless,sachdev2015bekenstein,kitaev2015simple,maldacena2016remarks} is an example of a $k$-local system with an extensive zero-temperature entropy, however it is unstable to perturbations by $2$-fermion terms.} Such a possibility violates Planck's formulation of the third law of thermodynamics in the strongest way possible \cite{wilks1961third,masanes2017general}, raising an intriguing question of whether any formulation of the third law can only be true for geometrically local systems.

A phase is generally characterized by a set of universal properties which are independent of the microscopic details of a typical system inside that phase.
A distinctive characteristic of geometrically local Hamiltonians which describe a topological phase of matter is that local perturbations cannot lift the ground state degeneracy or close the spectral gap,  up to corrections that vanish exponentially in the thermodynamic limit. This is remarkable, because for a generic local perturbation $V=\sum_i V_i$ where $V_i$ has some finite norm $J$ and acts on the vicinity of the $i$'th qubit,  $\norm{V}$ scales with the system size $N$ while the gap is a constant $O(1)$, but nonetheless, the perturbation $V$ cannot close the gap. This observation, which is known as the stability of topological Hamiltonians \cite{bravyi2010topological,bravyi2011short}, can be understood as a consequence of the fact that there exists a stable quantum field theory at long wavelengths that gives a well-defined description of the low-energy physics of geometrically local topological systems in the thermodynamic limit \cite{witten1989quantum,wen1990ground,wen04}. In this work, we investigate whether a similar notion of stability could exist for $k$-local Hamiltonians, where we apparently can no longer use quantum field theory as a guide.  

The robustness of the ground state degeneracy in a topological phase of matter can be understood heuristically as a consequence of local indistinguishability of the corresponding ground states \cite{kitaev2003fault}. This means that for an operator to be able to distinguish between different ground states it has to have a non-trivial support on $\mathcal{O}(L)$ qubits, where $L$ is the linear system size. As such, any local perturbation $V$ has zero matrix elements between different ground states up to $\mathcal{O}(L)$-th order in perturbation theory.
Hence, such a perturbation can only lift the ground state degeneracy by at most $\mathcal{O}(J^L)$, which goes to zero in the thermodynamic limit $L\to \infty $ for small $J$. This heuristic argument was made rigorous  in Refs. \cite{bravyi2010topological,bravyi2011short} for exactly solvable topological Hamiltonians on Euclidean lattices and subsequently generalized to include more generic topological Hamitonians \cite{michalakis2013stability,nachtergaele2019quasi,nachtergaele2022quasi}. 

Returning to the question of the stability of $k$-local Hamiltonians, similar reasoning as the heuristic argument above suggests the possibility of stable $k$-local phases of matter. First we note that in the absence of an underlying geometry, the notion of \textit{local} indistinguishability can be naturally replaced by the notion of low-weight indistinguishability, meaning that any operator that is supported on a constant number of qubits should not be able to distinguish between different ground states. Interestingly, the $k$-local Hamiltonians associated to qLDPC codes naturally satisfy low-weight indistinguishability. This is because the low-weight indistinguishability is nothing but the Knill-Laflamme condition for quantum error correction \cite{PhysRevLett.84.2525}. The same line of reasoning then seems to suggest that low-weight perturbations cannot lift the ground state degeneracy of qLDPC Hamiltonians in the thermodynamic limit. Nonetheless, there are subtleties to consider which are unique to the $k$-local problem. For example, the ground state degeneracy of $k$-local Hamiltonians is exponentially large if the Hamiltonians are associated to constant rate qLDPC codes. 
Therefore, even if the matrix elements of the effective perturbation in the ground-state subspace is exponentially suppressed, the ground-state degeneracy might still be lifted in the thermodynamic limit. As such, a more rigorous treatment of the stability question in the $k$-local case is called for.

In this work, we investigate to what degree the rigorous treatment of topological stability in  Ref. \cite{bravyi2011short} can be generalized to the $k$-local case. To this end we follow the same proof strategy as \cite{bravyi2011short}, but we relax the assumption of an underlying Euclidean geometry. Moreover, we keep track of the scaling of different parameters in order to be able to have quantitative estimates on the perturbation strength thresholds for stability in the thermodynamic limit. We show that the stability of a $k$-local Hamiltonian $H$ is closely related to the properties of its interaction graph $G$. In particular, we find that if there exist constants $\varepsilon_1, \varepsilon_2>0 $ such that the size of balls of radius $r$ on the interaction graph is upper bounded by $O(\exp(r^{1-\varepsilon_1}))$ and balls of radius $O(\log(n)^{1+\varepsilon_2})$ are locally indistinguishable (see Section \ref{sec_notation} for precise definitions), then the associated Hamiltonian is stable against local perturbations on the same graph. The more precise statement of our main result is presented in Theorem \ref{th_main}.

As a non-trivial example of a non-euclidean qLDPC code whose stability can be inferred from this theorem, we analyze the families of semi-Hyperbolic surface codes with code parameters $\llbracket N, \Theta(N^{1-\varepsilon}), \Theta(N^{\frac{\varepsilon}{2}}\,\log N)\rrbracket$ for a fixed $0<\varepsilon<1$. By tuning $\varepsilon$, This code can be tuned to be anywhere between the hyperbolic code (with $\Gamma(r)\sim \exp(r)$) and the toric code (with $\Gamma(r)\sim r^2$). Moreover, the thermodynamic limit $N \rightarrow \infty$ does not correspond to the infinite volume limit of a fixed manifold, so it cannot be understood through the lens of geometrically local topological phases. We show that all these Hamiltonians (excluding the hyperbolic code at $\varepsilon = 0$) have robust ground state degeneracy and robust gap above the ground state. Interestingly, we find that the perturbation strength threshold for the stability of the Hamiltonians in the thermodynamic limit approaches a constant independent of $\varepsilon$ as $\varepsilon \to 0$. 

The rest of the paper is structured as follows. In Section \ref{sec_notation} we introduce the notation and basic definitions that we are going to use throughout the paper. We also mention few general Lemmas which will be used in the course of the proof. Section \ref{sec_main_result} contains the rigorous statement of the theorem about the stability of the gap and robustness of the ground state degeneracy in \(k\)-local non-trivial Hamiltonians. The technical proof of this theorem is laid out in Section \ref{sec_proof}. In Section \ref{sec_example_semihyperbolic} we show how one can apply our theorem to infer the stability of semi-Hyperbolic surface code Hamiltonian. Section \ref{thirdlawsec} discusses our results in light of the third law of thermodynamics. Appendix \ref{apx_filterfunction} reviews the construction of filter functions which are used to define the quasi-adiabatic continuation operator. And lastly, in Appendix \ref{apx_LR} we review and prove many lemmas which are related to Lieb-Robinson type bounds and quasi-adiabatic continuation for \(k\)-local Hamiltonians.  


\section{Notation, definitions and basic lemmas}\label{sec_notation}
For a given graph $G$, let $B_u(r)$ denote the ball of radius $r$ around a site $u$,
\begin{align}
    B_u(r)=\{v\in G| d(v,u)\le r \},
\end{align}
where $d(u,v)$ is the graph distance. In an abuse of notation, here and throughout this work, we use $G$ to denote the graph as well as its set of vertices. 
We use $\Gamma(r)$ to denote an upper bound on the size of all balls of radius $r$, i.e.
\begin{align}
    |B_u(r)| \le \Gamma(r)\qquad \text{for all }u
\end{align}
For example, in a $d$ dimensional Euclidean lattice, $\Gamma(r)=c~r^d$, while in an expander graph, such as a uniform tiling of the hyperbolic plane,  $\Gamma(r)=c~\exp(\alpha~r)$ (see \cite[Fig.~3]{breuckmann2017} for a visualization).

The diameter of a subset $A$, denoted by $\text{diam}(A)$, is the maximum graph distance between any two points in $A$. We use $D_G$ to denote the diameter of the whole graph. For a given subset $A$ we define its $r$-neighborhood, denoted by $b_r(A)$, to be,
\begin{align}
    b_r(A)=\bigcup_{u\in A}B_u(r).
\end{align}

A local commuting projector Hamiltonian $H_0$ is specified by a graph $G$ and a set of commuting projectors $Q_u$ such that
\begin{align}\label{eq_local_projector_hamiltonian_form}
    H_0=\sum_{u \in G} Q_u,
\end{align}
and each $Q_u$ has only support on $B_u(\xi)$ for some constant range $\xi$. For simplicity, in the rest of this work we assume $\xi=1$, but it is straightforward to see all the results hold for any constant $\xi$ with minimal modifications. 
For later reference, we define the projections into the global and local ground state subspaces of $H_0$ as 
\begin{align}
    &P=\prod_{u\in G} (I-Q_u), \quad Q:=I-P\\
    &P_A=\prod_{\substack{u: B_u(1)\subseteq A}}
         (I-Q_u), \quad Q_A:=I-P_A,
\end{align}
where $I$ denotes the identity. 
$A$ could be any subset of the vertices but we mostly consider cases where $A=B_u(r)$ for some $u$ and $r$. 

We are interested in how the spectral gap of a local commuting projector Hamiltonian $H_0=\sum_u Q_u$ behaves upon perturbing $H_0$ by a perturbation $V=\sum_u V_u$. The \textbf{spectral gap} of $H=H_0+V$ is defined to be the energy difference between $(M+1)$-th and $M$-th lowest eigenstates of $H$, where $M$ is the ground state degeneracy of $H_0$. We need to impose some notion of locality on $V$, since without restricting the locality or the strength of $V_u$, the gap is clearly unstable. One may restrict $V_u$ to be local with respect to the interaction graph defined by $H_0$. While this is a reasonable choice, it makes the results inapplicable to interesting cases such as toric code under $k$-local perturbations.  Another choice then would be to just limit $V_u$ to be $k$-local. It turns out that the stability result we are able to prove is highly dependent on the interaction graph induced by $V$, and only assuming $k$-locality of $V$ is insufficient for proving our result. Therefore,  in this work we assume that the input to the problem includes a graph $G$ and we state the stability result with respect to perturbations which are local with respect to $G$. Locality with respect to a graph is defined as the following,

\begin{definition}\label{def_local_prtb_on_graph}\textbf{Local perturbation on a graph:} Given a graph $G$, we say the operator $V_u$ is \textit{localized} around $u$, has strength $J$ and tail $f(r)$ if one can write $V_u$ as 
\begin{equation}
V_u=\sum_{r\ge 1} V_{u,r}    
\end{equation}
such that $V_{u,r}$ is only supported on $B_{u}(r)$ -- with $B_u(r)$ denoting the ball of radius $r$ around $u$ on the graph $G$ -- and $\norm{V_{u,r}}\le J~f(r)$. We say the operator $V$ has strength $J$ and tail $f(r)$ if $V=\sum_{u\in G}V_u$ where $V_u$ is localized around $u$, has strength $J$ and tail $f(r)$. 
\end{definition}

In Definition \ref{def_local_prtb_on_graph} and throughout the text, $\norm{O}$ denotes the spectral norm of operator $O$. We require $H_0$ to be local with respect to $G$, but $G$ can have extra edges compared to the interaction graph of $H_0$, to allow for more general types of perturbations $V$, e.g. toric code with $k$-local perturbations. As we will show, the stability of the spectral gap of $H_0$ with respect to local perturbations on $G$ depends on how fast the size of neighborhoods in $G$ grow as a function of the distance.

Lastly, the proof utilizes the notion of globally block diagonal perturbations and locally block diagonal perturbations, which are defined as follows. 
\begin{definition}
\textbf{Globally block diagonal and locally block diagonal perturbations:}
Let $V=\sum_u V_u$ and $V_u=\sum_r V_{u,r}$, where $V_{u,r}$ is only supported on $B_u(r)$. We say $V$ is \textit{globally block diagonal} with respect to $H_0$ if $[V_u,P]=0$ for all $u$. We say $V$ is \textit{locally block diagonal} with respect to $H_0$ if $[V_{u,r},P]=0$ for all $u$ and $r$. 
\end{definition}

\subsection{Locality bounds for unitary evolution in general graphs}

The proof uses exact quasi-adiabatic continuation \cite{hastings2005quasiadiabatic,osborne2007simulating} ---  a type of unitary evolution --- to transform generic Hamiltonians to simpler forms. The following Lemmas are about the locality properties of the transformed operators. The proofs can be found in the Appendix \ref{apx_LR}. 
\begin{lemma}\label{lm_loc_decom_filter}
Let $H$ be a Hamiltonian of strength $J_1$ and tail $ e^{-\mu r}$. Let $F(t)$ be a real function defined on $\mathbb{R}$ and assume that for any $t\ge 0$,  $\int_{|t'|>t}\dd t' \,|F(t')|\le c_1\,\eta_a(t)$, for constants $c_1$ and $a$, where
\begin{align}\label{eq_eta_a_definition}
    \eta_a(x)=\exp \qty[-a \frac{x}{\ln^2(e^2+x)}].
\end{align}
Let $V_u$ be an operator localized around $u$, with strength $J$ and tail $f(r)$. Define $\widetilde{V}_u$  as,
\begin{align}
    \widetilde{V}_u=\int_{-\infty}^{\infty}\dd t\, F(t)\, e^{i\,H\,t}~V_u~e^{-i\,H\,t}.
\end{align}
Then $\widetilde{V}_u$ is localized around $u$, has strength $J$ and has tail $\tilde{f}(r)$ given as
\begin{align}
    \tilde{f}(r)=c\, r\, \Gamma(r) \max\{f(r/2),\eta_{a}(r/(4v))\},
\end{align}
for a constant $c$ which depends on $c_1$, $a$, $\mu$ and $v$, where $v$ is the following
\begin{align}
    v=\frac{8 J_1}{\mu}\sum_{r\ge 1}\Gamma(r)^2 e^{-\frac{\mu}{2}r}.
\end{align}
\end{lemma}

\begin{lemma}\label{lm_loc_decom_eta_unitary}
Let $V_u$ be an operator localized around $u$ with strength $J$ and tail $f(r)$. Let $H_s$ be a Hamiltonian of strength $J_1$ and tail $c_1\, r\, \Gamma(r)\, \eta_a(r/b)$ for all $s\in[0,1]$ for constants $a$, $b$ and $c_1$. Let $U_t$ be the evolution operator under $H_s$ from $s=0$ to $s=t\le 1$. Then $U_t^\dagger V_u U_t$ is localized around $u$, has strength $J$ and has tail $\tilde{f}$ given as,
\begin{align}
    \tilde{f}(r)=c\, r\, \Gamma(r)\max\{f(r/2),\eta_{a/2}(r/(4b))\},
\end{align}
for a constant $c$ which depends on $a$, $b$, $c_1$ and $v$, where $v$ is the following Lieb-Robinson \cite{lieb1972finite} velocity:
\begin{align}
    v=c_1\,J_1\,\sum_r r\Gamma(r)^{3}~ \eta_{a/2}(r/b).
\end{align}
\end{lemma}

\subsection{Indistinguishability of local ground states}

Local indistinguishability of ground states ensures that local perturbations can only lift the ground state degeneracy up to small corrections that vanish in the thermodynamic limit. However, to have a stable gap, one needs to also make sure that the energy shift of low-lying excited states in response to a local perturbation does not differ significantly from that of the ground states, since otherwise the energy of a low-lying excited state could drop below the energy of the perturbed ground states, thus closing the gap. One way to ensure this is to require that local ground states look the same as the global ground states. Note that the low-lying excited states of a commuting projector Hamiltonian satisfy $Q_u\ket{\psi}=0$ almost everywhere except at a few points where the projectors are violated, and hence they are local ground states of any patch of the Hamiltonian that does not contain those points. Local indistinguishability of local ground states from the global ground states then ensures that a local perturbation, for the most part, shifts their energy as much as it shifts the ground state energy and thus keeps the gap open. The following definition make this notion precise (see also Definition 4 of Ref.\cite{michalakis2013stability}):
\begin{definition}
\textbf{Locally indistinguishable region:}
Let $H_0=\sum_{u} Q_u$ be a local commuting projector Hamiltonian. A region $A$ is said to be \textit{locally indistinguishable} if for any operator $O$ solely supported on $A$ we have 
\begin{align}\label{eq_def_loc_indist}
    P_B O P_B=c P_B,
\end{align}
where $B=b_1(A)$ and $c$ is a complex number.
\end{definition}
Note that Eq.\eqref{eq_def_loc_indist} implies that the local ground states on $A$, i.e. states $\ket{\psi}$ that satisfy $Q_u\ket{\psi}=0$ for any $Q_u$ with non-trivial support on $A$, have the same reduced density matrix on $A$ as that of the global ground states, and hence are locally indistinguishable from them. This implies the following lemma, which has been proved in Ref.\cite{bravyi2011short} for Euclidean lattices. The same proof works on general graphs, and is included in Appendix \ref{apx_proof_of_lm_lcgc} for completeness. 
\begin{lemma}\label{lm_lcgc}
Let $A$ be a locally indistinguishable region, and let $B$ be its 1-neighborhood, i.e. $B=b_1(A)$. Let $O$ be an operator supported on $A$. Then for any region $C$ enclosing $B$, i.e. $B\subseteq C$, we have
    \begin{align}
        \norm{O P}=\norm{O P_C}
    \end{align}
\end{lemma}
Lastly, we define the local indistinguishability radius to be a measure of how large a region should be so it can support operators which could distinguish different local ground states. 
\begin{definition}
\textbf{Local indistinguishability radius:}
For a given local commuting projector Hamiltonian $H_0$, the \textit{local indistinguishability radius} $\rho^\ast$ is the maximum radius such that $B_u(r)$ is locally indistinguishable for all $u$ and all $r<\rho^\ast$. 
\end{definition}

\begin{eg}\label{ex_stabilizer_code_hams}
\textbf{Stabilizer code Hamiltonians and their interaction graph:}
Consider a Pauli stabilizer code $\mathcal{C}$ on $N$ qubits, specified by a set of code stabilizers $\{g_1,\cdots,g_m\}$ that generate the stabilizer group of the code space $\mathcal{G}=\langle g_1, g_2, \cdots, g_m \rangle$, with $m\le N$. We use ``code stabilizers'' to specifically reference the generators $g_i$ of the corresponding stabilizer group.  The interaction graph $G_\mathcal{C}$ of the code $\mathcal{C}$ given the code stabilisers $\{g_1,\cdots,g_m\}$ is defined as follows: The vertices correspond to the set of qubits and there is an edge between two qubits if there is a code stabilizer $g_i$ that acts non-trivially on both of them. Note that $G_\mathcal{C}$ depends on the choice of the generating set $\{g_1,\cdots,g_m\}$, and is not a feature of the code $\mathcal{C}$ alone.
By definition all code stabilizers are strictly local on the interaction graph with range $\xi=1$. The next step is to assign each code stabilizer to a vertex of the graph $G_\mathcal{C}$. It is straightforward to show using induction that it is always possible to assign each code stabilizer $g_i$ to one of the vertices on its support such that no more than one stabilizer is assigned to each vertex. Thus we may view the code space as the ground state subspace of the the local commuting projector Hamiltonian as expressed in Eq.\eqref{eq_local_projector_hamiltonian_form} with $Q_u=(1-g_i)/2$. One may take the graph $G$ -- with respect to which the perturbations are local -- to be the same as $G_\mathcal{C}$, or let $G$  be $G_\mathcal{C}$ but with extra edges to consider more general perturbations. In general the stability of the spectral gap depends on the choice of $G$.  A region $A$ is locally indistinguishable if and only if a) it does not contain a logical operator, and b) the elements of the stabilizer group $\mathcal{G}$ which are solely supported on $A$ can all be generated by stabilizers $g_i$ that are solely supported on $B=b_1(A)$ (see Lemma 2.1 in Ref.\cite{bravyi2010topological}). It is clear that a) and b) are necessary conditions for $A$ to be a locally indistinguishable region. To show they are sufficient, first suppose $O$ is a Pauli operator supported on $A$. If $O$ anti-commutes with any $g_i$ that is contained in $B$, then $P_B O P_B=0=0\, P_B$. If $O$ commutes with every $g_i$ contained in $B$, then $O$ commutes with all the stabilizers in $\mathcal{G}$, because it commutes trivially with the rest of code stabilizers $g_i$ due to their disconnected support. Since $O$ cannot be a logical operator, it should be a stabilizer up to a sign which means $P_B\,O P_B \propto P_B$, since all stabilizers on $A$ can be generated by code stabilizers $g_i$ on $B$. The claim follows by noting that any operator can be written as a sum over Pauli operators. 
Note that the local indistinguishability radius $\rho^\ast$ depends on $G$. In toric code and its variants and when $G$ is taken to be $G_\mathcal{C}$, it is easy to see that b) is always satisfied for simply connected regions, and hence $\rho^\ast \sim d/2$, where $d$ is the code distance. 
\end{eg}

\subsection{Energy Spectrum for relatively bounded perturbations}

The following lemma is the Lemma 4 in Ref.\cite{bravyi2011short}, which is valid for any Hamiltonian irrespective of the underlying geometry. We will make use of it in the last part of the proof. 
\begin{lemma}\label{lm_rlt_bnd_pert}
Let $W$ be a relatively bounded perturbation to the Hamiltonian $H_0$, i.e. there exists a $b>0$ such that,
\begin{align}
    \norm{W \psi}\le b\norm{H_0 \psi}, \quad\text{for any }\ket{\psi}.
\end{align}
Then the spectrum of $H_0+W$ is contained in the union of intervals $[(1-b)\lambda_0,(1+b)\lambda_0]$ where $\lambda_0$ runs over eigenvalues of $H_0$. 
\end{lemma}


\section{Main Result}\label{sec_main_result}

Our goal is to prove the stability of the spectral gap and the robustness of the ground state degeneracy for a certain class of interaction graphs $G$. The proof strategy relies on 1) transforming a generic perturbation first to a globally block diagonal perturbation via exact quasi-adiabatic continuation, and 2) transforming it further into a locally block diagonal perturbation by re-summing the terms that are supported on small balls (compared to the indistinguishability radius) and treating the rest (terms that are supported on large balls) as errors. 
Steps 1-2 above will not change the energy spectrum because the  quasi-adiabatic continuation in the first step is a unitary operation and the total norm of the errors in the second step vanishes in the thermodynamic limit. However the tail of the perturbation will change under these transformations. 
If one starts with a perturbation of strength $J$ and tail $e^{-\mu r}$, the first step turns it into a globally block diagonal perturbation of strength $O(J)$ and tail
\begin{align}\label{eq_modified_tail_f}
f(r)=r^5~\Gamma(r)^7~\eta_{\frac{1}{28}}\left(\frac{r}{256~v_1}\right),
\end{align}
where $\eta_a(x)$ is defined in Eq.\eqref{eq_eta_a_definition} and $v_1$ is the Lieb-Robinson velocity which bounds the following sum,
\begin{align}\label{eq_th_v1}
    \frac{8(e^\mu+1)}{\mu} \sum_{r\ge 1}^{D} \Gamma(r)^2 e^{-\frac{\mu}{2}r} \le v_1<\infty.
\end{align}
The resummation in the second step will result in a locally block diagonal perturbation of strength $O(J)$ and tail $\bar{f}(r)$ defined as,
\begin{align}\label{eq_modified_tail_fbar}
    \bar{f}(r)=\sum_{r'\ge r}^{D} f(r').
\end{align}
The gap and ground state degeneracy is robust if this tail function decays rapidly enough, the precise meaning of which is outlined in the following.
\begin{theorem}\label{th_main}
Consider a family of local commuting projector Hamiltonians $\{H_n \}_{n=1}^\infty$ of increasing size $N_n$. $H_n$ is defined on the graph $G_n$ with Gamma function $\Gamma_n(r)$ and diameter $D_n$ and has local indistinguishability radius $\rho^\ast_n$. Consider perturbing the Hamiltonian $H_n$ by local perturbation $V$ of strength $J$ and tail $\exp(-\mu r)$. Let $f_n(r)$ and $\bar{f}_n(r)$ denote the modified tail functions as defined in Eq.\eqref{eq_modified_tail_f} and Eq.\eqref{eq_modified_tail_fbar}, using $D=D_n$ and $\Gamma=\Gamma_n$, and assume there exists a constant (independent of $n$) $v_1$ that bounds the sum in Eq.\eqref{eq_th_v1} for large enough $n$. If there exists a constant $b_0$ such that for large enough $n$,
\begin{align}\label{eq_th_assumption}
    \sum_{r\ge 1}^{D_n} \Gamma_n(r+1)^{\frac{1}{2}}\bar{f}_n(r)\le b_0<\infty,
\end{align}
and if
\begin{align}\label{eq_th_error}
    \lim_{n\to\infty} N_n~D_n~\bar{f}_n(\rho^\ast_n)=0,
\end{align}
then, there exist constants $n_0$ and $b$ (with $b$ finite and depending only on $b_0$, $\mu$ and $v_1$) such that for $n>n_0$ and $J<J_0=1/(4b)$, the spectrum of $H_n+s\,V_n$ for $s\in[0,1]$ is contained in the union of
\begin{align}\label{eq_I_k}
    I_k=[k(1-bJ)-\delta_n,k(1+bJ)+\delta_n ], \quad k \in \mathbb{Z}_{\geq 0}
\end{align}
up to a constant, where $\delta_n=O(N_n D_n \bar{f}_n(\rho^\ast_n))\to 0$ as $n\to\infty$. In particular, the ground states are in $I_0$ so the ground state degeneracy is restored in the thermodynamic limit, and the spectral gap $\Delta \geq 1 - bJ > 3/4$ 
in the limit $n \rightarrow \infty$.
\end{theorem}

\begin{corollary}
If $\Gamma_n(r)=O(\exp(r^{1-\varepsilon_1}))$ and $\rho^\ast_n=\Omega(\log(N_n)^{1+\varepsilon_2})$ for constants $\varepsilon_1,  \varepsilon_2 > 0$, then the gap and the ground state degeneracy is robust against weak local perturbations. 
\end{corollary}
\begin{proof}
Clearly there exists a constant $v_1$ bounding the sum in Eq.\eqref{eq_th_v1} independent of $n$. Moreover, it is easy to see from Eq.\eqref{eq_eta_a_definition} that for large $r$, $\eta_a(r)$ decays faster than $\exp(-r^\alpha)$ for any $\alpha<1$. Choose $\alpha$ such that $\max(1-\varepsilon_1,\frac{1}{1+\varepsilon_2})<\alpha<1$, and choose a large enough constant $C$ such that $\eta_{\frac{1}{28}}(r)\le C \exp(-r^{\alpha})$ for all $r$. Given that $\Gamma_n(r)=O(\exp(r^{1-\varepsilon_1}))$ and that $\alpha>1-\varepsilon_1$, one can see from Eq.\eqref{eq_modified_tail_f} that for the modified tail function $f_n$ we have $f_n(r)=O(\exp(-c\, r^\alpha))$ for some constant $c>0$. This in turn ensures that $\bar{f}_n(r)=O(r\exp(-c\, r^{\alpha}))=O(\exp(-c'\,r^\alpha))$ for a constant $c'>0$. Since $\alpha>1-\varepsilon_1$, the sum in Eq.\eqref{eq_th_assumption} is clearly bounded by a constant $b_0<\infty$. Moreover, since $\rho_n^\ast=\Omega(\log(N_n)^{1+\varepsilon_2})$ and $\alpha > \frac{1}{1+\varepsilon_2}$, $\bar{f}_n(\rho^\ast_n)$ decays faster than any polynomial in $N_n$, thus the limit in Eq.\eqref{eq_th_error} vanishes. The statement of the Corollary then follows from Theorem \ref{th_main}.
\end{proof}

\section{Proof of Theorem \ref{th_main}}\label{sec_proof}
\subsection{Assuming a finite gap}\label{subsec_gap}

First, we prove that it suffices to prove Theorem \ref{th_main} under the assumption that the spectral gap of $H_n + s\, V_n$, denoted by $\Delta_n(s)$, is lower bounded by $1/2$ for large enough $n$. The argument follows exactly Ref.\cite{bravyi2011short}, which we repeat here for completeness. The proof is by contradiction. Assume we have proved the theorem under the assumption $\Delta_n(s)\ge 1/2$. Now assume there is a family of Hamiltonian $H_n$ and a series of perturbations $V_n$, which satisfy the conditions of the theorem but for any finite $J$ the gap drops below $1/2$ at some point for arbitrary large $n$. Continuity of $\Delta_n(s)$ and the fact that $\Delta_n(0)=1$ implies that there should exists $s^\ast\in[0,1)$ where $\Delta_n(s^\ast)=1/2$, and that  $\Delta_n(s)\ge 1/2$ for $s\in [0,s^\ast]$. $s^\ast$ could be a function of $n$ and $J$. Let us choose $J=J_0/2$,  where $J_0$ is specified in the statement of the theorem. Then, the statement of the theorem should be valid for $s\in[0,s^\ast]$ given the fact that we have proved the theorem for $\Delta\ge 1/2$. However, this means that $\Delta_n(s^\ast)\ge 3/4$ for large enough $n$, which contradicts $\Delta_n(s^\ast)=1/2$.  This shows that we may assume  $\Delta_n(s)\ge 1/2$ without loss of generality. 

\subsection{Reduction of general local perturbations to globally block diagonal perturbations}

In this section we prove that as far as the spectrum of  $H_s=H_0+s\, \sum_u V_u=H_0 + s\, V$ is concerned, one may assume $[P,V_u]=0$, at the cost of $V$ having a slightly different tail. The intuition behind this is as follows. 

Note that $Q_u$ stabilizes the ground space of $H_0$, but not that of $H_s=H_0+sV$. However, assuming $H_s$ is in the same phase as $H_0$, we expect there exists slightly different dressed operators $\widetilde{Q}_u$ which stabilize the ground space of $H_s$. In particular $[P(s),\widetilde{Q}_u]=0$, where $P(s)$ is the projection onto the ground space of $H_s$. 

Now, let us rewrite $H_s$ as $H_s=\sum_u \widetilde{Q}_u +\widetilde{V}$ where $\widetilde{V}=H_s-\sum_u \widetilde{Q}_u=\sum_u(Q_u-\widetilde{Q}_u)+s\,V$, and note that $[P(s),\widetilde{V}]=[P(s),H_s-\sum_u \widetilde{Q}_u]=0$. This shows that $H_s$ is a sum of commuting projectors $\widetilde{Q}_u$ perturbed by perturbation $\widetilde{V}$ which commutes with the projection onto its ground space $P(s)$. 

Let $U_s$ denote the quasi adiabatic continuation unitary, which relates the dressed operators to bare operators. Using $U_s$, we can undress $\widetilde{Q}_u$ back to $Q_u$ and $P(s)$ back to $P$. Doing so turns $H_s$ into $H'_s=U^\dagger_s H_s U_s=\sum_u U^\dagger_s \widetilde{Q}_u U_s + U^\dagger_s \widetilde{V} U_s=H_0 + V'$, where now we have $[P,V']=0$. Since $U_s$ is unitary, the spectrum of $H'_s$ is the same as $H_s$. 

The last step is to write $V'$ in terms of local operators $V'=\sum_u V'_u$ such that each local term commutes with $P$, i.e. $[P,V'_u]=0$ for all $u$. This step follows from the standard procedure used to rewrite a gapped local Hamiltonian as a frustration free local Hamiltonian where the ground state is the eigen state of each term in the Hamiltonian\cite{hastings2006solving,kitaev2006anyons}. 

Putting it all together, this shows that as far as the spectrum of $H_s$ is concerned, we may instead study the spectrum of $H'_s$ which is $H_0$ perturbed by a globally block diagonal perturbation $V'$. It remains to relate the locality properties of $V'$ to that of $V$, which constitutes the main part of the rest of this section.

Consider the following class of Hamiltonians parameterized by $s\in [0,1]$
\begin{align}\label{eq_THE_hamiltonian}
    H_s=H_0+s\, V,
\end{align}
where $V$ is a local perturbation with strength $J$ and tail $e^{-\mu r}$, i.e.
\begin{equation}
    V=\sum_{u\in G}V_u=\sum_{u\in G}\sum_{r\ge 1} V_{u,r},\quad \norm{V_{u,r}}\le J e^{-\mu r}
\end{equation}
In general, the ground state subspace of $H_s$ is different from the ground state subspace of $H_0$ and the perturbation $V$ preserves neither one. The first step is to use  quasi-adiabatic continuation to transform $V$ into $V'$ such that $[V',P]=0$. Since quasi-adiabatic continuation is unitary, the spectrum remains the same and we can analyze the transformed Hamiltonian. The exact quasi-adiabatic evolution Hamiltonian $\mathcal{D}_s$ is defined as
\begin{align}\label{eq_exact_qac_definition}
    \mathcal{D}_s=\int_{-\infty}^{+\infty} \dd t ~W_{1/2}(t) ~\exp(i H_s t)(\partial_s H_s)\exp(-i H_s t),
\end{align}
where $W_\gamma(t)$ is a filter function defined in Appendix \ref{apx_filterfunction}, and the quasi adiabatic continuation operator $U_s$ is defined as the unitary evolution generated by $\mathcal{D}_s$,
\begin{align}\label{eq_qu_adiabatic_ev_U_def}
    U_s=\mathcal{S}'\qty[\exp(-i\int_0^{s}\dd s' \mathcal{D}_{s'})].
\end{align}
For simplicity, we drop the $1/2$ subscript of the filter function $W_{1/2}(t)$ and denote it simply by $W(t)$.  $W(t)$ is real so $\mathcal{D}_s$ is Hermitian. Moreover $\int_{|t'|>t} \dd t' |W(t')| \le c~\eta_{\frac{1}{14}}(t)$ for a constant $c$ for any $t \ge 0$ which in turn ensures $\mathcal{D}_s$ to be a quasi-local Hamiltonian, i.e. comprised of local terms with $\sim\eta(r)$ tail that decays faster than any polynomial but slower than a exponential, according to Lemma \ref{lm_loc_decom_filter}. Most importantly, $W(t)$ is defined such that $\widetilde{W}(\omega)=-\frac{i}{\omega}$ for $|\omega|>1/2$, where $\widetilde{W}$ is the Fourier transform of $W$; this ensures that $U_s$ would map ground states to exact ground states along the adiabatic path\cite{osborne2007simulating}. This is made precise in the following lemma (See Lemma 7.1 in Ref.\cite{bravyi2010topological} for its proof)
\begin{lemma}\label{lm_projection_mapping}
Let $H_s$ be a differentiable family of Hamiltonians. Let $\ket{\psi_i(s)}$ denote the eigen-states of $H_s$ with energy $E_i$. Let $E_\text{min}(s)<E_\text{max}(s)$ be two continuous functions of $s$. Define $P(s)$ as,
\begin{align}\label{eq_def_P_s}
    P(s)=\sum_{i: E_\text{min}(s)<E_i<E_\text{max}(s)} \ketbra{\psi_i(s)}.
\end{align}
assume that $[E_\text{min}(s),E_\text{max}(s)]$ is separated from the rest of the spectrum by a gap of at least $1/2$ for all $0\le s\le 1$. Then we have
\begin{align}
   U_s^\dagger\, P(s)\,U_s =P(0)
\end{align}
\end{lemma}
We choose $E_\text{min}(s)$ to be slightly below the lowest energy of of $H(s)$ and $E_\text{max}(s)$ to be slightly above the $M$-th lowest energy of $H(s)$ where $M$ is the ground state degeneracy of $H_0$. Therefore, $P(s)$ as defined in Eq.\eqref{eq_def_P_s} would be the projection onto the low-energy sector of $H(s)$, i.e. the perturbed ground space. In particular, since $P(0)=P$ we find,
\begin{align}\label{eq_projections_under_qac}
     U_s^\dagger\, P(s)\,U_s =P.
\end{align}

An important point to note is that Lemma \ref{lm_projection_mapping} requires the gap to remain open, which is why we had to argue in Section \ref{subsec_gap} that we may assume $\Delta_n(s)\ge 1/2$. Since $[H_s,P(s)]=0$, Lemma \ref{lm_projection_mapping} implies that  $[U_s^\dagger H_s U_s ,P]=0$, where $P$ is the projection on the ground state subspace of the unperturbed Hamiltonian $H_0$. Next, we write $U_s^\dagger H_s U_s$ as
\begin{align}\label{eq_Hps}
    H'_s=U_s^\dagger H_s U_s=H_0 + V',
\end{align}
where 
\begin{align}\label{eq_V_prime_def}
    V'=U_s^\dagger H_0 U_s - H_0 + s U_s^\dagger V U_s.
\end{align}
Importantly we have, 
\begin{align}
    [V',P]=0.
\end{align}
In what follows, we show that $V'$ is a quasi-local perturbation with strength $J$ and we compute its tail function. First note that by definition,
\begin{align}
    \mathcal{D}_s=\sum_{u} \int_{-\infty}^{+\infty} \dd t ~W(t) ~\exp(i H_s t)V_{u}\exp(-i H_s t).
\end{align}
Also note that $H_s=H_0+sV$ has strength $(e^\mu+sJ)$ and tail $e^{-\mu r}$. It is convenient to assume $J<1$, so we may upper bound the strength of the Hamiltonian $H_s$ to be $(e^\mu+1)$, which is independent of $s$ and $J$. Therefore, without loss of generality, in the rest of the proof we assume $J<1$. Based on Lemma \ref{lm_loc_decom_filter},  $\mathcal{D}_s$  has strength $J$ and tail $c\, r\, \Gamma(r) \max\{e^{-\mu r}, \eta_{1/14}(r/(4v_1))\}$, where $v_1$ is given as,
\begin{align}\label{eq_LR_velocity_1}
      v_1=\frac{8(e^\mu+1)}{\mu} \sum_r \Gamma(r)^2 e^{-\frac{\mu}{2}r},
\end{align}
which will be finite in the thermodynamic limit if the assumptions of Theorem \ref{th_main} hold.
Moreover, one can choose a constant $c'$ such that $c'\,\eta_{1/14}(r/(4v_1))\ge e^{-\mu r}$ for all $r\ge 0$. So we can take the tail of $\mathcal{D}_s$ to be,
\begin{align}\label{eq_Ds_tail}
    g(r)=c_\mathcal{D}\, r\, \Gamma(r) \,\eta_{1/14}(r/(4v_1)),
\end{align}
where $c_\mathcal{D}$ is a constant that depends on $\mu$ and $v_1$. In the rest of this section, we will use the same reasoning to simplify the tail expressions in Lemma \ref{lm_loc_decom_filter} and Lemma \ref{lm_loc_decom_eta_unitary} whenever possible. 

Now we can use quasi-locality of $\mathcal{D}_s$ to infer the  locality properties of  $V'$ defined in Eq.\eqref{eq_V_prime_def}. Based on Lemma \ref{lm_loc_decom_eta_unitary}, $U_s^\dagger V U_s$ has strength $J$ and tail 
\begin{align}
    h_1(r)=c_1~r~\Gamma(r)~\eta_{1/28}(r/(16v_1)),
\end{align}
where $c_1$ is a constant that depends on $\mu$, $v_1$ and $v_2$, where $v_2$ is given as,
\begin{align}\label{eq_LR_velocity_2}
    v_2= c_\mathcal{D}\,J\sum_r\,r\,\Gamma^3(r)\, \eta_{\frac{1}{28}}(r/(4v_1)),
\end{align}
which remains finite in the thermodynamic limit as well if the assumptions of Theorem \ref{th_main} hold. As for the first two terms $U_s^\dagger H_0 U_s - H_0$, we note that
\begin{align}\label{eq_uhu_h}
    U_s^\dagger H_0 U_s - H_0=\int_0^s\dd s' \partial_{s'}\qty(U_{s'}^\dagger H_0 U_{s'})=-i \int_0^s \dd s' U_{s'}^\dagger [\mathcal{D}_{s'},H_0] U_{s'},
\end{align}
where we have used $\partial_s U_s=i \mathcal{D}_s U_s$ to get to the third expression. Since  $\mathcal{D}_s$ has strength $J$ and tail $g(r)$ (given by Eq.\eqref{eq_Ds_tail}), one can write it as $\mathcal{D}_s=\sum_r D_{u,r}$ with $\norm{D_{u,r}}\le J g(r)$. 
Therefore, we may write the commutator above as $[\mathcal{D}_s,H_0]=\sum_{u,r} \sum_{u'}[D_{u,r},Q_{u'}]=\sum_{u,r\ge 3}A_{u,r}$, where
\begin{align}
    A_{u,r}=\sum_{u':~d(u,u')\le r}[D_{u,r-1},Q_{u'}].
\end{align}
Note that $A_{u,r}$ is supported on $B_u(r)$ and, 
\begin{align}
    \norm{A_{u,r}}\le& \sum_{u':~d(u,u')\le r}2\norm{D_{u,r-1}}\norm{Q_{u'}}\le 2J~ \Gamma(r)\,g(r-1)\\
    &\le c_Q~J~r~\Gamma(r)^2~ \eta_{\frac{1}{14}}(r/(4v_1)),\label{eq_DH0_loc_decom}
\end{align}
for a constant $c_Q$ that depends on $\mu$ and $v_1$. This shows $[\mathcal{D}_s,H_0]$ has strength $J$ and tail $ c_Q~J~r~\Gamma(r)^2~ \eta_{\frac{1}{14}}(r/(4v_1))$. Then, by using Lemma \ref{lm_loc_decom_eta_unitary}, one finds that $ U_{s'}^\dagger [\mathcal{D}_{s'},H_0] U_{s'}$ has strength $J$ and tail 
\begin{align}
    h_2(r)=c_2~r^2~\Gamma(r)^3~ \eta_{1/28}(r/(16v_1)).
\end{align}
$c_2$ is a constant that depends on $\mu$, $v_1$ and $v_2$. Since this bound is uniform in $s'$, it remains valid for $ \int_0^1 \dd s' U_{s'}^\dagger [\mathcal{D}_{s'},H_0] U_{s'}$.

Therefore, by combining the local decompositions for $U_s^\dagger V U_s$ and $U_s^\dagger H_0 U_s-H_0$, we see that $V'$ (given by Eq.\eqref{eq_V_prime_def}) has strength $J$ and tail
\begin{align}\label{eq_vprime_tail}
    h(r)=c_{V'} \, r^2\, \Gamma(r)^3\,\eta_{\frac{1}{28}}(r/(16 v_1)),
\end{align}
for some constant $c_{V'}$ which depends on $\mu$, $v_1$ and $v_2$.

Finally, we need another transformation so each term $V'_u=\sum_r V'_{u,r}$ commutes with $P$ (So far it is only $V'=\sum_u V'_u$ which commutes with $P$).
To this end, we use the filter function $w_{\frac{1}{2}}(t)$. The family of $w_\gamma(t)$ filter functions is briefly reviewed in Appendix \ref{apx_filterfunction}. For simplicity, for the rest of the paper we denote $w_{1/2}(t)$ simply by $w(t)$. The filter function $w(t)$ is related to the previous filter function $W(t)$ via the following,
\begin{align}\label{eq_W_from_w}
    W(t_1)=
    \begin{cases}
      \int_{t_1}^\infty \dd t~w(t),& t_1\ge 0,\\
      -\int_{-t_1}^\infty \dd t~w(t),& t_1< 0.
    \end{cases}
\end{align}
In particular, $w(t)$ has the following properties (see Appendix \ref{apx_filterfunction} for details): $w(t)$ is even and non-negative. Also, $\int_{t}^\infty w(t)\le c\, \eta_{\frac{1}{14}}(t)$ for a constant $c$ for any $t\ge 0$, and hence it satisfies conditions of Lemma \ref{lm_loc_decom_filter}.
Moreover, it is chosen such that  $\widetilde{w}(0)=1$ and $\widetilde{w}(\omega)=0$ for $\omega> 1/2$, where $\widetilde{w}$  is the Fourier transform of $w(t)$.

Since $\widetilde{w}(0)=1$, we may rewrite $H'_s$ (given in Eq.\eqref{eq_Hps}) as,
\begin{align}
    H'_s=\int_{-\infty}^{+\infty} \dd t~w(t)~\exp(i H_s' t) H'_s \exp(-i H_s' t)=H_0 + \sum_u (\widetilde{Q_u}-Q_u)+\sum_u \widetilde{V'_u}.
\end{align}
where 
\begin{align}\label{eq_V_tilde}
    \widetilde{V'_u}=\int_{-\infty}^{+\infty} \dd t~w(t) \exp(i H_s' t) V'_u  \exp(-i H_s' t).
\end{align}
and $V'_u=\sum_r V'_{u,r}$. $\widetilde{Q_u}$ is defined similarly. On the other hand, since $\widetilde{w}(\omega)=0$ for $\omega >1/2$, it is easy to eigendecompose the above and see that $\mel{E_j}{\widetilde{V'_u}}{E_i}=0$ for $|E_i-E_j|> 1/2$, where $\ket{E_i}$ denotes the eigen states of $H_s'$ with energy $E_i$. Since $H_s'$ is related to $H_s$ by a unitary transformation, its spectral gap is also lower bounded by $1/2$. More importantly, due to Eq.\eqref{eq_projections_under_qac}, the projection $P$ is the projection on the ground state subspace of $H'_s$ as well. Hence, $(1-P)\widetilde{V'_u} P=P\widetilde{V'_u}(1-P)=0 $. The same is true about $\widetilde{Q_u}-Q_u$. This means that $[\widetilde{V'_u},P]=[\widetilde{Q_u}-Q_u,P]=0$. It remains to show that they have strength $J$ and rapidly decaying tails.

 To argue about locality of $\widetilde{V'_u}$, first we rewrite Eq.\eqref{eq_V_tilde} as,
 \begin{align}\label{eq_u_vtilde_udagger}
     U_s \widetilde{V'_u} U_s^\dagger=\int_{-\infty}^{+\infty} \dd t~w(t)~\exp(i H_s t) U_s V'_u U_s^\dagger \exp(-i H_s t).
 \end{align}
where now the time evolution inside the integral is generated by $H_s$ (without the prime) which has exponentially vanishing tails. 
Based on Lemma \ref{lm_loc_decom_eta_unitary}, $U_s V'_u U_s^\dagger$ is localized around $u$, has strength $J$ and decay function $c_3~r^3~\Gamma(r)^4~\eta_{\frac{1}{28}}(r/(32 v_1))$ for a constant $c_3$ which depends on $\mu$, $v_1$ and $v_2$. Then, based on Lemma \ref{lm_loc_decom_filter},  $U_s \widetilde{V'_u} U_s^\dagger$ is localized around $u$, has strength $J$ and tail $c'_3~ r^4 \Gamma(r)^5~\eta_{1/28}(r/(64 v_1)$ for a constant $c'_3$ which depends on $\mu$, $v_1$ and $v_2$. Finally, since  $\widetilde{V'_u}$ is the evolution of  $U_s \widetilde{V'_u} U_s^\dagger$ under $U_s$, based on Lemma \ref{lm_loc_decom_eta_unitary} it is localized around $u$, has strength $J$ and tail $c''_3~r^5 ~\Gamma(r)^6~\eta_{1/28}(r/(128 v_1))$ for a constant $c''_3$ which depends on $\mu$, $v_1$ and $v_2$.

As for $\widetilde{Q_u}-Q_u$, note that, 
\begin{align}
    \widetilde{Q_u}&=\int_{-\infty}^{+\infty} \dd t~w(t)~\exp(i H_s' t) Q_u  \exp(-i H_s' t)\\
    &=\int_{-\infty}^{+\infty} \dd t~w(t)\Bigg[Q_u+\int_0^t\dd t_1~ \frac{\partial}{\partial t_1}\Bigg(\exp(i H_s' t_1) Q_u  \exp(-i H_s' t_1)\Bigg)\Bigg]\\
                    &= Q_u +i \int_{-\infty}^{+\infty} \dd t~ w(t)~\int_0^t \dd t_1 \exp(i H_s' t_1) [V',Q_u]  \exp(-i H_s' t_1)\\
                    &= Q_u +i \int_{-\infty}^{+\infty} \dd t_1~W(t_1)~\exp(i H_s' t_1) [V',Q_u]  \exp(- i H_s' t_1),\label{eq_qtilde_minus_q}
\end{align}
where in the last line we have used Eq.\eqref{eq_W_from_w}. We may write $[V',Q_u]$ as the following,
\begin{align}
    [V',Q_u]=\sum_{u',r}[V'_{u',r},Q_u]=\sum_{r} \sum_{u':~d(u,u')\le r+1} [V'_{u',r},Q_u]=\sum_{r} Q'_{u,2r+1},
\end{align}
with $Q'_{u,2r+1}$ defined as $\sum_{u':~d(u,u')\le r+1} [V'_{u',r},Q_u]$. Note that $Q'_{u,2r+1}$ is supported on $B_u(2r+1)$ and its norm is upper bounded by,
\begin{align}
\norm{Q'_{u,2r+1}}\le \sum_{u':~d(u,u')\le r+1} \norm{[V'_{u',r},Q_u]}\le 2\Gamma(r+1) h(r),
\end{align}
with $h(r)$  given by Eq.\eqref{eq_vprime_tail}. This shows that $[V',Q_u]$ has strength $J$ and tail $c_4r^2\Gamma(r)^4\eta_{1/28}(r/(32 v_1))$ with $c_4$ a constant that depends on $\mu$, $v_1$ and $v_2$. Sandwiching both sides of \eqref{eq_qtilde_minus_q} by $U_s$ yields, 
\begin{align}
    U_s~(\widetilde{Q_u}- Q_u)~U_s^\dagger&= i \int_{-\infty}^{+\infty} \dd t_1~W(t_1)~\exp(i H_s t_1)~U_s~[V',Q_u]~U_s^\dagger~\exp(- H_s t_1).
\end{align}

The same line of reasoning as what followed Eq.\eqref{eq_u_vtilde_udagger} shows that $\widetilde{Q_u}-Q_u$ has strength $J$ and tail $c'_4~r^5 ~\Gamma(r)^7~\eta_{1/28}(r/(256\,v_1))$ with $c'_4$ a constant which depends on $\mu$, $v_1$ and $v_2$. By combining all contributions, we see that $W_u\equiv \widetilde{Q_u}-Q_u+\widetilde{V'_u} $ is localized around $u$, has strength $J$ and has tail
\begin{align}\label{eq_W_tail}
    f(r)=c_W~r^5 ~\Gamma(r)^7~\eta_{\frac{1}{28}}(\frac{r}{256 v_1}),
\end{align}
 where $c_W$ is a constant which depends on $\mu$, $v_1$ and $v_2$. Importantly, $[W_u,P]=0$.


\subsection{Reduction of globally diagonal perturbation to a locally diagonal one.}
So far we have shown that the energy spectrum of the Hamiltonian $H_0+V$ where $V$ has strength $J$ and tail $e^{-\mu r}$ is equivalent to energy spectrum of $H_0+W$,  where $W=\sum_u W_u$ is globally block diagonal, i.e. $[W_u,P]=0$ for all $u$, has strength $J$ and has tail $f(r)$ given in Eq.\eqref{eq_W_tail}.  In this section we show that one can rewrite $W$ as  $W=\widetilde{W}+\delta_1 +\text{const.}$, where $\widetilde{W}$ is locally block diagonal, has strength $J$ and has a vanishing tail which we will compute, and $\delta_1$ is an error term whose norm, $\norm{\delta_1}$, goes to zero in the thermodynamic limit.

Consider the interaction $W_u=\sum_r W_{u,r}$. Let $\rho^\ast$ denote the local indistinguishability radius of $H_0$. By adding the appropriate constants to $W_{u,r}$, we may assume that $P W_{u,r} P=0$ for $r< \rho^\ast$. Define $\delta_W$ as 
\begin{align}
    \delta_W=PW_u=PW_u P=\sum_{r \ge \rho^\ast}PW_{u,r}P.
\end{align}
Note that,
\begin{align}
    \norm{\delta_W}\le J\sum_{r \ge \rho^\ast}^D f(r)=J \bar{f}(\rho^\ast),
\end{align}
where $\bar{f}(r)$ is defined as,
\begin{equation}
    \bar{f}(r)\equiv \sum_{r'\ge r}^D f(r').
\end{equation}
Since $f(r)$ decays faster than any polynomial, $\bar{f}(r)$ is also a decaying function of $r$, vanishing faster than any polynomial. Define $W'_{u}=W_u-\delta_W$. Furthermore, we define  $W'_{u,r}=W_{u,r}$ for $r<D$ and $W'_{u,D}=W_{u,D}-\delta_W$, so $W'_u=\sum_r W'_{u,r}$, with $\norm{W'_{u,r}}< J \tilde{f}(r)$, with
\begin{equation}
\tilde{f}(r)=
    \begin{cases} 
      f(r) & r< D \\
      f(D)+\bar{f}(\rho^\ast) & r=D
   \end{cases}.
\end{equation}
Note that $PW'_u=W'_u P=0$, and therefore $W'_u$ has non-zero elements only between excited states, i.e. 
\begin{align}\label{eq_Wp_excited_prj}
     W'_u=(1-P)W'_u(1-P).
\end{align}
Let us define $E_{u,i}$ as follows,
\begin{align}\label{eq_E_expr}
    E_{u,i}= (1-P_{B_u(i)})P_{B_u(i-1)}=P_{B_u(i-1)}-P_{B_u(i)},
\end{align}
with $P_{B_u(0)}=I$. $E_{u,i}$ is the projection into the subspace spanned by the eigenstates of $H_0$ with at least a defect on $B_u(i)$ but no defect on $B_u(i-1)$, with having a \textit{defect} on some region \(A\) meaning that there is a site $v \in A$ for which $Q_v\ket{\psi}=\ket{\psi}$. 
In other words, it is the projection onto the subspace where the first defect appears in distance $i$ from $u$. 
 It is clear from Eq.\eqref{eq_E_expr} that $E_{u,i}$ is supported on $B_u(i)$. Furthermore, the \(E\)'s telescopically sum up to the projection outside of the codespace,
\begin{align}
    1-P=\sum_{i=1}^D E_{u,i}~.
\end{align}
By using Eq.\eqref{eq_Wp_excited_prj}, we may write $W'_u$ as,
\begin{align}
    W'_u =\sum_r W'_{u,r}=\sum_{i,j,r} E_{u,i} W'_{u,r} E_{u,j}.
\end{align}

Now each term in the sum commutes with $P$ and is supported on $B_u(R)$ with $R=\max(i,j,r)$. It remains to show that the norm decays faster than any power of $R$. To this end, we break the sum into three parts,
\begin{align}
    \sum_{i,j,r} E_{u,i} W'_{u,r} E_{u,j} =&\sum_r \sum_{i,j \le r} E_{u,i} W'_{u,r} E_{u,j} \nonumber \\
                                &+\sum_i \sum_{j,r< i} E_{u,i} W'_{u,r} E_{u,j} + E_{u,j} W'_{u,r} E_{u,i} \nonumber\\
                                &+\sum_i \sum_{r < i} E_{u,i} W'_{u,r} E_{u,i} \nonumber \\
                                =&\sum_r X_{u,r}  
                                +\sum_i Y_{u,i}  
                                +\sum_i Z_{u,i},\label{eq_XYZ}
\end{align}
with,
\begin{align}
    X_{u,r}=&\qty(\sum_{i\le r} E_{u,i}) W'_{u,r}\qty(\sum_{j\le r} E_{u,j})=(1-P_{B_u(r)})W'_{u,r}(1-P_{B_u(r)}),\\
    Y_{u,i}=&P_{B_u(i-1)}(1-P_{B_u(i)})\qty(\sum_{r<i} W'_{u,r}) (1-P_{B_u(i-1)})\nonumber\\
            &+(1-P_{B_u(i-1)})\qty(\sum_{r<i} W'_{u,r}) (1-P_{B_u(i)})P_{B_u(i-1)}\\
    Z_{u,i}=&P_{B_u(i-1)}(1-P_{B_u(i)})\qty(\sum_{r<i} W'_{u,r})(1-P_{B_u(i)})P_{B_u(i-1)}.
\end{align}
Note that $X_{u,r}$ is supported on $B_u(r)$ and $Y_{u,i}$ and $Z_{u,i}$ are supported on $B_u(i)$. Moreover, $X_{u,r}$, $Y_{u,i}$ and $Z_{u,i}$ are all Hermitian.

For $r<D$, we have $\norm{X_{u,r}}\le J f(r)$ because $\norm{W'_{u,r}}\le J f(r)$. We will treat $X_{u,D}$ as error $\delta_X$, with $\norm{\delta_X}\le J~(f(D)+\bar{f}(\rho^\ast))\le 2J~\bar{f}(\rho^\ast)$. As for $Y_{u,i}$, we show that the norm of the terms with $i\le \rho^\ast$ decays with $i$ and  we treat the terms with $i>\rho^\ast$ as errors with vanishing norm. First, note that for $i \le\rho^\ast$, we have
\begin{align}
    \frac{1}{2}\norm{Y_{u,i}}\le \norm{P_{B_{u}(i-1)} \sum_{r<i}W'_{u,r}}&\le \norm{P_{B_{u}(i-1)} \sum_{r\le i-2}W'_{u,r}}+\norm{W'_{u,i-1}}\\
    &\le \norm{P \sum_{r\le i-2}W'_{u,r}}+J\tilde{f}(i-1)\\
    &=\norm{PW'_u-P\sum_{r\ge i-1}W_{u,r}}+J\tilde{f}(i-1)\\
    &=\norm{P\sum_{r \ge i-1}W'_{u,r}}+J\tilde{f}(i-1) \\
    &\le J\sum_{r\ge i-1}\tilde{f}(r)+J\tilde{f}(i-1) \\
    &\le 2J\sum_{r\ge i-1}\tilde{f}(r)\\
    &= 2J \sum_{r\ge i-1} f(r) + 2J \bar{f}(\rho^\ast)=2J\,\bar{f}(i-1)+2J\,\bar{f}(\rho^\ast)\\
    &\le 4J\,\bar{f}(i-1).
\end{align}
In the first line we have used the fact that for any operator $A$ and projection $\Pi$, one has $\norm{A~\Pi}\le \norm{A}$ and that $\norm{A}=\norm{A^\dagger}$. To go from the first line to the second we have used Lemma \ref{lm_lcgc}, or more precisely the conjugate transposed version of that.
On the other hand if $i> \rho^\ast$, we have
\begin{align}
    \frac{1}{2}\norm{Y_{u,i}} &\le \norm{P_{B_{u}(i-1)}\sum_{r<i}W'_{u,r}}\le \norm{P_{B_{u}(i-1)}\sum_{r\le \rho^\ast-1}W'_{u,r}}+\norm{P_{B_{u}(i-1)}\sum_{\rho^\ast\le r<i} W'_{u,r}}\\
    &\le \norm{P\sum_{r\le \rho^\ast-1}W'_{u,r}}+\norm{\sum_{\rho^\ast\le r<i} W'_{u,r}}\le \norm{PW'_u -P\sum_{r\ge \rho^\ast}W'_{u,r}} + J \sum_{r\ge\rho^\ast} \tilde{f}(r)\\
    &\le 2J\,\bar{f}(\rho^\ast) + 2J \bar{f}(\rho^\ast)\le 4J~\bar{f}(\rho^\ast).
\end{align}
We define $\delta_Y$ as $\delta_Y=\sum_{i> \rho^\ast} Y_{u,i}$, and the above analysis shows that 
\begin{align}
    \norm{\delta_Y}\le 8~J~D~\bar{f}(\rho^\ast).
\end{align}
Almost the same calculation shows that $\norm{Z_{u,i}}\le 4J~\bar{f}(i-1)$ for $i\le \rho^\ast$. Moreover, we define $\delta_Z=\sum_{i>\rho^\ast} Z_{u,i}$ and the same reasoning as above shows that  $\norm{\delta_Z}\le 4~J~D~\bar{f}(\rho^\ast)$.

Putting everything together, we see that
\begin{align}
    W_u=W'_u +\delta_W+\text{const.}=\widetilde{W}_u+\delta_u+c,
\end{align}
where $\widetilde{W}_u$ is given as,
\begin{align}
    \widetilde{W}_u=\sum_r^{D-1} X_{u,r}+\sum_{i=1}^{\rho^\ast} Y_{u,i}+\sum_{j=1}^{\rho^\ast} Z_{u,j},
\end{align}
where $X_{u,r}$, $Y_{u,r}$ and $Z_{u,r}$ are Hermitian and $[X_{u,r},P]=[Y_{u,r},P]=[Z_{u,r},P]=0$. Hence, $\widetilde{W}=\sum_u \widetilde{W}_u$ is locally block diagonal, has strength $J$ and has a tail,
\begin{align}
    f_{\widetilde{W}}(r)=c_{\widetilde{W}}\sum_{r'\ge r-1} r'^5~\Gamma(r')^7~\eta_{\frac{1}{28}}(\frac{r'}{256\,v_1}),
\end{align}
where $c_{\widetilde{W}}$ is constant which depends on $\mu$, $v_1$ and $v_2$. Also $\delta_u$ is given as,
\begin{align}
    \delta_u=\delta_W+\delta_X+\delta_Y+\delta_Z,
\end{align}
whose norm is upper bounded by $ O\qty(D~\bar{f}(\rho^\ast))$. Since this error is only for a single site, The norm of the total error is bounded by $O\qty(N~D~\bar{f}(\rho^\ast))$, which goes to zero based on the assumption of the theorem.


\subsection{Locally block diagonal perturbations are relatively bounded}

In this section we show that locally block diagonal Hamiltonians are relatively bounded by $H_0$, up to small errors which go to zero in the thermodynamic limit.  
We start by partitioning the Hilbert space using the projections $Q_u$. Recall that we say a state $\ket{\psi}$ has a defect on site $u$ if $Q_u\ket{\psi}=\ket{\psi}$. Note that the ground states have no defects. Let $\bm{\lambda}=(\lambda_1,\cdots,\lambda_N)$ be a vector in $\mathbb{Z}_2^N$. We define the projection $P_{\bm{\lambda}}$ as,
\begin{align}
    P_{\bm{\lambda}}=\prod_{u:\lambda_u=0}(I-Q_u)\prod_{u:\lambda_u=1}Q_u.
\end{align}
Basically, $\bm{\lambda}$ specifies the location of the defects, and $ P_{\bm{\lambda}}$ is the projection on the corresponding subspace. Clearly, the set of $P_{\bm{\lambda}}$ projectors make a complete set of orthogonal projectors,
\begin{align}
    \sum_{\bm{\lambda}}P_{\bm{\lambda}}=I,\text{ and} \quad P_{\bm{\lambda}}P_{\bm{\lambda}'}=0\quad \text{for } \bm{\lambda}\ne \bm{\lambda}'.
\end{align}
Also note that $H_0=\sum_{\bm{\lambda}} |\bm{\lambda}| P_{\bm{\lambda}}$, where $|\bm{\lambda}|$ stands for the weight of $\bm{\lambda}$.

Consider the locally block diagonal perturbation $\widetilde{W}=\sum_{u,r}\widetilde{W}_{u,r}$. By adding appropriate constants we may assume $P \widetilde{W}_{u,r}P=0$ for $r< \rho^\ast$, with $\rho^\ast$ denoting the local indistinguishability radius. Since $\widetilde{W}$ is locally block diagonal, it means $\widetilde{W}_{u,r}P=0$ for $r< \rho^\ast$. Then by Lemma \ref{lm_lcgc}, we see that for $r< \rho^\ast$ we have $\widetilde{W}_{u,r}P_B=0$, where $B=b_1(B_u(r))=B_u(r+1)$ is the one-neighborhood of \(B_u(r)\). Therefore, if $\bm{\lambda}$ does not have any defect in $B$, then $\widetilde{W}_{u,r} P_{\bm{\lambda}}=0$ because $P_{\bm{\lambda}}=P_B\times\cdots $ where the ellipsis represent the rest of projections in $P_{\bm{\lambda}}$. As we will see this can be used to show that up to vanishing errors, $\widetilde{W}$ is relatively bounded by $H_0$.   

Let $\ket{\psi}$ be any wave function. Let $\widetilde{W}_{r}=\sum_u\widetilde{W}_{u,r}$, and assume $r< \rho^\ast$. Then, 
\begin{align}
    \norm{\widetilde{W}_r\psi}^2&=\norm{\sum_u \widetilde{W}_{u,r} \sum_{\bm{\lambda}} P_{\bm{\lambda}} \psi}^2\le \sum_{\bm{\lambda}} \sum_u\norm{ \widetilde{W}_{u,r} P_{\bm{\lambda}} \psi}^2\\
                &\le \sum_{\bm{\lambda}} \sum_u\norm{ \widetilde{W}_{u,r} P_{\bm{\lambda}}}^2\norm{P_{\bm{\lambda}} \psi}^2 \le \sum_{\bm{\lambda}} |\bm{\lambda}|~\Gamma(r+1)~J^2 f^2_{\widetilde{W}}(r) \norm{P_{\bm{\lambda}} \psi}^2\\
                &\le \Gamma(r+1)~J^2 f^2_{\widetilde{W}}(r)~\sum_{\bm{\lambda}}|\bm{\lambda}|^2 \norm{P_{\bm{\lambda}}\psi}^2= \Gamma(r+1)~J^2 f^2_{\widetilde{W}}(r)\mel{\psi}{\sum_{\bm{\lambda}} |\bm{\lambda}|^2  P_{\bm{\lambda}}}{\psi}\\
                &= \Gamma(r+1)~J^2 f^2_{\widetilde{W}}(r)\norm{H_0 \psi}^2~,
\end{align}
where in the second line we used the fact that $\widetilde{W}_{u,r}P_{\bm{\lambda}}=0$ if $\bm{\lambda}$ does not have any defect in $b_1(B_u(r))$. To be more precise, only for sites within  $r+1$ distance of a defect in $\bm{\lambda}$ we have $\widetilde{W}_{u,r}P_{\bm{\lambda}}\ne 0$ and there are at most $|\bm{\lambda}|\Gamma(r+1)$ number of such sites $u$. By summing over $r$ up to $\rho^\ast$ we get,
\begin{align}
    \norm{\sum_{r<\rho^\ast} \widetilde{W}_r\psi}\le J \qty(\sum_{r<\rho^\ast} \Gamma(r+1)^{1/2} f_{\widetilde{W}}(r))\norm{H_0 \psi}.
\end{align}
In other words, $\sum_{r<\rho^\ast} W_r$ is relatively bounded by $H_0$ with parameter $J b$, where $b$ is an upper bound on the following sum, 
\begin{align}
    \sum_{r<\rho^\ast} \Gamma(r+1)^{1/2} f_{\widetilde{W}}(r)<b,
\end{align}
which remains finite based on the assumptions of the theorem (see \ref{eq_th_assumption}). 

Finally, we may write $\widetilde{W}$ as follows,
\begin{align}
    \widetilde{W}=\sum_{r< \rho^\ast} \widetilde{W}_r + \sum_{r\ge \rho^\ast} \widetilde{W}_r.
\end{align}
The first is part is relatively bounded. The second part can be treated as an error $\delta_{\widetilde{W}}$, whose norm, 
\begin{align}
    \norm{\delta_{\widetilde{W}}}\le\sum_{r\ge\rho^\ast} \norm{W_r}\le N~D~J f_{\widetilde{W}}(\rho^\ast),
\end{align}
vanishes in the thermodynamic limit. Then the statement of Theorem \ref{th_main} follows from Lemma \ref{lm_rlt_bnd_pert}.

\section{Example: Semi-hyperbolic Surface Code}\label{sec_example_semihyperbolic}
The \eczoohref[2d hyperbolic surface code]{two_dimensional_hyperbolic_surface} \cite{freedman2002z2,breuckmann2016} is based on tiling of closed 2d hyperbolic surfaces with regular polygons. A specific tiling can be described by a set of two numbers $\{p,q\}$,
known as Schl\"{a}fli symbols, which represents a tiling of the plane with regular $p$-sided polygons such that $q$ of them meet at every vertex. When $1/p+1/q<1/2$, the associated tiling can only be realized on a surface with negative curvature.  

Given a $\{p,q\}$ tiling, one can define a stabilizer code where physical qubits lie on the edges and each vertex (plaquette) represents a
$Z$-type($X$-type) stabilizer. A topologically non-trivial closed hyperbolic surface with $g$ handles has $2g$ non-trivial independent
loops which can be used to define $k=2g$ logical qubits that are stabilized by the code. On the other hand, the surface area of a closed hyperbolic surface (with constant negative curvature) is proportional to $g-1$ (due to the Gauss-Bonnet theorem), and because the surface is tiled by regular polygons, it means that the number of physical qubits $N$ would be also proportional to $g-1$. Therefore, if one considers a family of closed hyperbolic surfaces with increasing number of handles $g$ that admit a $\{p,q\}$-tiling, one can construct a family of qLDPC codes which has a constant encoding rate $r=k/N=\Theta(1)$. The distance $d$ grows logarithmically with the number of qubits, resulting in a code $\mathcal{Q}_H=\llbracket N, r\,N, \gamma\,\log N \rrbracket$ for constants $r,\gamma >0$. As a consequence of the hyperbolic geometry, if one considers the interaction graph associated with these codes, one finds $\Gamma_n(r)\sim \exp(\alpha\,r)$ with a positive constant $\alpha$ for $r$ much smaller than the code distance $d_n$. Given that the balls with radius less than $d_n/2$ are topologically trivial, they do not contain any logical operator, thus $\rho^\ast_n\sim \log N$.

There is a variant of the hyperbolic code, known as the semi-hyperbolic code \cite{breuckmann2017}, where one considers a series of the hyperbolic surface code  associated to $\{p=4,q\}$ tilling, and then replace all squares by a $l \times l$ square lattice. This would result in a qLDPC code with parameters $\mathcal{Q}_\text{SH}=\llbracket n\,l^2, r\,n, \gamma\,l\,\log n \rrbracket$. If we let $l$ to be a function of $n$, we can tune between the hyperbolic surface code and the normal toric code on surfaces with non-trivial topology by controlling how $l_n$ scales with $n$. 

We consider the case where $l_n=n^a$ for a fixed $a\ge 0$. It results in a semi-hyperbolic surface code family with parameteres $\mathcal{Q}_\text{SH}=\llbracket n^{1+2a}, r\,n, \gamma~n^{a}~\log n \rrbracket=\llbracket N, r\,N^{1-\varepsilon}, (1-\varepsilon)\,\gamma\,\sqrt{N^\varepsilon}\,\log N \rrbracket$, with $0 \le \varepsilon=\frac{2a}{1+2a} <1$. As for the interaction graph associated to the semi-hyperbolic code, we have
\begin{align}
\Gamma_n(r)=
\begin{cases} 
      \alpha\, r^2 & r\le l_n \\
      l_n^2\, e^{\beta\,r/l_n} & l_n<r\le D_n
  \end{cases},
\end{align}
where $D_n=O(l_n \log(n))$ is the diameter of the interaction graph, with constants $\alpha $ and $\beta$ which depend on the tiling parameter $q$.  Moreover, arguments similar to the hyperbolic surface code show that $\rho^\ast_n\sim l_n\log(n)\sim \sqrt{N}^\varepsilon \log(N)$ for the semi-hyperbolic surface code. As we will show below, this family of semi-hyperbolic surface code for $a>0$ (i.e., $\varepsilon>0$) satisfies the assumptions of Theorem \ref{th_main} and as such, results in a stable qLDPC Hamiltonian. Moreover, we argue that the perturbation strength threshold $J_0$ is a constant independent of $a>0$.

We start by upper bounding the function $\bar{f}_n(r)$. We are going to estimate the discrete sums with integrals. For simplicity we consider perturbations with exponential tails $e^{-\mu r}$ whose parameter $\mu$ is large enough such that $v_1$ in Eq.\eqref{eq_th_v1} is bounded even for $a=0$. Note that the same value of $v_1$ bounds the sum in Eq.\eqref{eq_th_v1} for other values of $a\ne 0$, since a positive $a>0$ simply makes $\Gamma_n(r)$ grow slower. Next, we choose a constant $c$ (again independent of $a$) such that $\eta_{1/28}(r/(256 v_1))\le c\, \exp(-\sqrt{r})$ for all $r\ge 0$. Then for $r\le l_n$, we have,
\begin{align}
    \bar{f}_n(r)&\le c_1 \int_r^{l_n} \dd r' r'^{19} e^{-\sqrt{r'}}
    +c_2\, l_n^{14}\int_{l_n}^{\gamma\, l_n \log(n)}  \dd r' r'^5 e^{7\beta r'/l_n}e^{-\sqrt{r'}}\\
    &\le c_1 \int_r^{\infty} \dd r' r'^{19} e^{-\sqrt{r'}}
    +c_2\, l_n^{20}\, n^{7\gamma \beta}e^{-\sqrt{l_n}}\int_{1}^{\gamma\, \log(n)}  \dd u u^5 \\
    &\le \text{poly}(r)e^{-\sqrt{r}}+\text{poly}(n)e^{-n^{a/2}}.\label{eq_fbar_shc} 
\end{align}
Importantly, the first term in \eqref{eq_fbar_shc} does not depend on $a$. For $l_n<r\le D_n$, we have only the second integral above so,
\begin{align}
    \bar{f}_n(r)\le \text{poly}(n)e^{-n^{a/2}}.
\end{align}
First note that for any $a> 0$,
\begin{align}
    \lim_{n\to\infty} N_n D_n \bar{f}_n(\rho^\ast_n)\le\lim_{n\to\infty} \text{poly(n)}e^{-n^{a/2}}=0.
\end{align}
Therefore, Eq.\eqref{eq_th_error} holds for any $a>0$.
Moreover, it is easy to see that the sum in Eq.\eqref{eq_th_assumption} is bounded by a constant independent of $a>0$:
\begin{align}
    \int_0^{D_n}\dd r \Gamma_n(r+1)^{1/2}\bar{f}_n(r)&\le \int_0^{l_n-1} \dd r~(r+1)\,(\text{poly}(r)\, e^{-\sqrt{r}}+\text{poly}(n)e^{-n^{a/2}})\\
    &\quad+l_n\,\int_{l_n-1}^{\gamma\, l_n\, \log(n)}\dd r\, e^{\beta\, (r+1)/(2l_n)} \text{poly}(n)e^{-n^{a/2}}\\
    &\le \int_0^{\infty} \dd r\,\text{poly}(r)\, e^{-\sqrt{r}}+\text{poly(n)}e^{-n^{a/2}}
\end{align}
Therefore, by choosing $n_0$ large enough, we can bound the sum in Eq.\eqref{eq_th_assumption} by a constant $b_0$ which is independent of $a>0$ for $n\ge n_0$. Hence, there is a constant $J_0$ (independent of $a$, or equivalently independent of $\varepsilon=\frac{2a}{1+2a}$), such that the semi-hyperbolic surface code Hamiltonian is stable in the thermodynamic limit, if the perturbation strength $J$ is less than $J_0$. 

\subsection{Comparison to decoupled stacks of 2d codes}

It is worth comparing the semi-hyperbolic family of code Hamiltonians to another trivial construction with similar code parameters which is inspired by topological Hamiltonians with \eczoohref[fracton]{fracton} order \cite{nandkishore2019fractons}. Consider  stacking a number $L^{d-2}$ of $2$-dimensional toric codes of length $L$ into a hypercubic lattice of length $L$ in $d$ dimensions, which we take to be the interaction graph\footnote{To be more precise, the interaction graph associated to the stack of toric codes as defined in Example \ref{ex_stabilizer_code_hams} is just a stack of decoupled toric code interaction graphs. However, given that they are stacked in $d$ dimensions, the natural choice for the graph underlying local perturbations is the $d$ dimensional hypercube.}. Such a system encodes $k=L^{d-2}$ logical qubits in $N=L^d$ physical qubits and has code distance $L$, and hence has code parameters $\llbracket N, N^{1-\varepsilon}, \sqrt{N^\varepsilon}\rrbracket$ with $0<\varepsilon=\frac{2}{d}\le 1$. Since the interaction graph is a hypercubic lattice, we have $\Gamma(r)\sim r^d$ and $\rho^\ast\sim N^{1/d}$. As such, this family of codes also gives rise to a family of stable Hamiltonians with code parameters which are similar to that of the semi-hyperbolic surface code. However, in contrast to the semi-hyperbolic surface code, the perturbation strength threshold $J_0$ defined in Theorem \ref{th_main} is no longer constant as a function of $\varepsilon$, but vanishes as one takes $\varepsilon$ to $0$ (equivalently, \(d \to \infty\)). This is due to the fact that both the LR velocity $v_1$ in \eqref{eq_th_v1} and $b_0$ in \eqref{eq_th_assumption} diverge as one makes $\varepsilon$ arbitrary small (or, equivalently, $d$ arbitrary large). 


\section{Relation to third law of thermodynamics}
\label{thirdlawsec}

The third law of thermodynamics \cite{wilks1961third,masanes2017general}, due originally to Nernst in 1906 \cite{nernst1906ueber}, has several different formulations. The version that concerns us here, which is due to Max Planck, is that the entropy of a pure substance approaches zero when the temperature approaches absolute zero. In quantum systems, this is known not to hold due to the possibility of ground state degeneracies. However, it is interesting to consider whether there is a reformulation of this statement which is true. 

First, a natural refinement is to demand that the non-zero ground state entropy must be robust to generic perturbations. In this case, the only quantum systems that are known to host robust non-zero ground state entropies are topologically ordered states of matter \cite{wen1990ground}. These have the property that the robust ground state entropy is a finite constant. By considering layered systems, as described in the previous section, one can then construct examples of quantum systems with robust ground state entropies that are infinite in the thermodynamic limit. However there is still one final potential reformulation, which is whether there are any limits on how the robust ground state entropy scales with system size. 

In the last section, we showed that the naive case of layered 2d topologically ordered states, along with closely related fracton orders, have a ground state entropy in $d$-dimensions that scales as $k \sim N^{1 -2/d}$. However we do not have a stability result as $d \rightarrow \infty$, because the threshold $J_0$ in Theorem \ref{th_main} goes to zero. The results here give a significantly stronger example by proving that the semi-hyperbolic codes give examples of quantum systems where the ground state entropy scales as $k \sim N^{1 - \epsilon}$, for any $\epsilon > 0$, and are robust to a finite perturbation strength. 

This raises the question: is it possible for a quantum system to have an extensive ground state entropy $k \sim N$ which is robust to perturbations up to some finite threshold? If not, perhaps this would provide a proper reformulation of Planck's original statement. On the other hand, if it is possible to have a robust, extensive ground state entropy, then perhaps the correct reformulation of Planck's statement is that a robust, extensive ground state degeneracy is not possible in \it geometrically local \rm quantum systems. Here geometrically local is defined by considering the system on a fixed Riemannian manifold and taking the thermodynamic limit of infinite volume. Given all that is known about topological phases of matter, it appears unlikely that a robust, extensive ground state entropy can be possible in a geometrically local system. Indeed, Haah recently proved that under some conditions, geometrically local topologically ordered states have ground state entropy bounded by $O(N^{1-2/d})$ \cite{haah2021}. 

The above discussion suggests there is likely a proper reformulation of Planck's original statement, and we leave it as an open question to determine which of the possibilities above must hold. 

\section{Discussion and outlook}

In this work we investigated to what extent the assumption of an underlying Euclidean geometry in the proof technique of Ref. \cite{bravyi2011short} can be relaxed to make it applicable to \(k\)-local (i.e., not geometrically local but still sparse) Hamiltonians. Our result can be seen as pushing the previous result for the stability of the topological Hamiltonians on Euclidean  graphs with $\Gamma(r)=O(\text{poly}(r))$ to general graphs with $\Gamma(r)=O(\exp(r^{1-\varepsilon}))$,  where $\varepsilon$ can be any positive number. It is worth emphasizing that our theorem did not use the $k$-local property explicitly, but rather the expansion properties of the interaction graph, through which $k$-locality enters implicitly. 

Our result is an almost exponential improvement, and also demonstrates that a stable limit can exist where $\epsilon \rightarrow 0$ while the stability threshold for the perturbation strength remains finite. Nevertheless, unfortunately this falls just short of including the most interesting examples, such as the hyperbolic surface code \cite{breuckmann2016} or known constructions of asymptotically good sparse error-correcting codes \cite{panteleev2022asymptotically,leverrier2022quantum,Dinur2022,lin2022good}. In all of these examples, the neighborhood size grows exponentially with distance, so $\epsilon = 0$. 
It is worth noting that the conditions laid out in Theorem \ref{th_main} are sufficient for stability but most probably not necessary. As such, the question of stability of the aforementioned Hamiltonians remains an open question.

The main reason that our result is limited to interaction graphs with sub-exponential neighborhood growth is the use of \textit{exact} quasi-adiabatic continuation to convert the generic perturbations to globally diagonal perturbations. An exact quasi-adiabatic continuation would necessarily transform a local perturbation to quasi-local perturbations with tails that decay slower than an exponential. If the size of the neighborhoods in the underlying graph grow exponentially with distance, such quasi-local interactions can lead to non-local effects, as is suggested by the divergence of the Lieb-Robinson velocity in Lemma \ref{lm_LR_bound_for_eta_tail}. This can be seen to be analogous to the fact that there is no Lieb-Robinson velocity for power-law interactions $1/r^\alpha$ with exponent $\alpha<2d$ on $d$-dimensional Euclidean lattices \cite{tran2021lieb}.

In principle, one could try to use approximate quasi-adiabatic continuation \cite{hastings2005quasiadiabatic} to avoid dealing with quasi-local perturbations. For the  approximate quasi-adiabatic continuation, one needs to replace the filter function $W_{1/2}(t)$ in Eq.\eqref{eq_exact_qac_definition} with another filter function, whose tail decays exponentially or faster, but whose Fourier transform is only approximately equal to $-i/\omega$ for $\omega > 1/2$. One usually considers a  parameter-dependent family of filter functions, where one can control the error of the approximation by increasing the tuning parameter (see e.g. Ref. \cite{hastings2015quantization}). However, doing so would result in worse bounds on the norm of $\mathcal{D}_s$ as was originally pointed out in Ref. \cite{bravyi2010topological}. As such, while naively transforming perturbations under approximate quasi-adiabatic transformations does result in globally block diagonal perturbations with exponential tails, their strength diverges as one tunes the filter function parameter to improve the approximation. It would be interesting to find a way of utilizing approximate quasi-adiabatic continuation that would result in reasonable bounds on the strength of the transformed perturbations. It is worth emphasizing that the same problem arises even in Euclidean lattices and is not related to the peculiar connectivity of interaction graphs in \(k\)-local Hamiltonians.

One may also try to pursue other strategies to prove the stability of \(k\)-local phases of matter. In particular, it has been recently shown in Refs.~\cite{rakovszky2023physics,rakovszky2024physics} that qLDPC codes can be viewed as the result of gauging certain classical codes, and as such one may draw on intuitions about gauge theories to argue about the stability of the corresponding Hamiltonian. Whether this line of reasoning can be made rigorous remains unclear.  

Moreover, there are other notions of stability related to qLDPC codes, whose exploration might prove to be insightful to the spectral gap stability problem.  For example, the stability of eigenstate many-body localization under generic perturbations has recently been proven for a class of Hamiltonians related to good classical LDPC codes \cite{Yin24}. The ground state gap of these models is unstable to perturbations, but this result proves a type of rigidity in the spectral statistics at finite energy density.  If similar behavior occurs in Hamiltonians based on quantum LDPC codes, then it could potentially be exploited in proving stability of the spectral gap.  A central barrier in our proof for expander graphs is the difficulty in bounding the rapid spreading of perturbations that might be suppressed in models with stable eigenstate localization.  As another example, it is known that the stability of topological phases of matter is closely related to the fact that topological states cannot be disentangled by a constant depth local unitary circuit \cite{chen2010local}. On the other hand, it is straightforward to prove that code words of any quantum error correcting code with non-zero number of logical qubits and diverging distance cannot be disentangled to a product state by a constant depth \(k\)-local unitary circuit \cite{bravyi2006lieb,aharonov2013guest} (cf. ~\cite{stephen2024nonlocal}).
In the case of finite temperature transitions, it was recently shown that a sharp jump in the depth of the required disentangling circuit is \textit{not} a universal characteristic of a thermodynamic phase \cite{Hong24}.  
It is worth investigating whether a concrete connection can be made between the stability of the spectral gap and the circuit depth of the disentangling circuits.

Furthermore, it was shown in Ref.~\cite{anshu2023nlts} that the \(k\)-local Hamiltonians of certain qLDPC codes satisfy the ``No Low-Energy Trivial State'' (NLTS) property \cite{freedman2013quantum}, which means that any quantum state with small enough energy density under the aforementioned Hamiltonian is non-trivial, i.e. cannot be disentangled by a constant depth local unitary circuit. It was further argued in Ref.~\cite{chen2023symmetry} that the ground state subspace of such Hamiltonians is thermally stable. Therefore, one wonders whether the NLTS property can be used to prove the stability of the spectral gap as well. 

One may also build insights by tackling the spectral gap stability problem on models that lie in between the full-blown \(k\)-local case studied here and the known-to-be-stable geometrically local case. For example it might be insightful to investigate the stability of the toric code Hamiltonian under $k$-local perturbations. Stability results related to gapped geometrically local Hamiltonian with long-range interactions, such as power-law interactions, could be especially relevant here \cite{michalakis2013stability,Lapa_2023}.

Lastly, qLDPC codes possess a finite error threshold against various local noise channels \cite{leverrier2015quantum,fawzi2018efficient,panteleev2021degenerate,breuckmann2021single,gu2023efficient} and as such, it is worth exploring whether one can relate the stability of quantum information encoded in the ground state subspace against local noise to the stability of the spectral gap under \(k\)-local perturbations (see e.g. Refs.~\cite{lavasani2024stability,yaodong2024}).

\section{Acknowledgments}

AL thanks Spiros Michalakis, Jeongwan Haah and Tarun Grover for insightful discussions.
This research was supported by NSF QLCI grant
OMA-2120757 and Grant No.\ NSF PHY-1748958, NSF DMR-2345644, an NSF CAREER grant (DMR-1753240), and the Laboratory for Physical Sciences (LPS) through the Condensed Matter Theory Center (CMTC) at UMD.


\bibliographystyle{unsrt}
\bibliography{refs}

\newpage
\appendix
\section{The filter function $w_\gamma(t)$}\label{apx_filterfunction}
In this section we briefly review the construction and bounds of the filter function $w_\gamma$ introduced in Ref.\cite{bachmann2012automorphic}. Here, we only state the results. The interested reader could consult Ref.\cite{bachmann2012automorphic} for the proofs.

For a given gap parameter $\gamma>0$, define the sequence $a_n=\frac{a_1}{n\ln^2 n}$ and set $a_1$ such that $\sum_{n=1}^\infty a_n=\gamma/2$. Define the function $w_\gamma(t)$ as,
\begin{align}
    w_\gamma(t)=c_\gamma\prod_{n=1}^\infty\qty(\frac{\sin a_n t}{a_n t})^2,
\end{align}
where $c_\gamma$ is a constant which we choose such  that $\int w_\gamma(t) \dd t=1$.  The function $w_\gamma(t)$ has the following properties:
\begin{itemize}
    \item $w_\gamma$ is an even, non-negative function in $L^1(\mathbb{R})$  (meaning that $\int_{\mathbb{R}}\dd t |w_\gamma(t)|<\infty$), with $w_\gamma(t)\le 1$ for all $t$
    \item $\widetilde{w}_\gamma(0)=1$, where $\widetilde{w}_\gamma$ is the Fourier transform of $w$. 
    \item $\widetilde{w}_\gamma(\omega)=0$ for $|\omega|\ge \gamma $.
    \item $w(t)$ decays in $t$ faster than any power law. More precisely,  we have the following bound
    \begin{align}\label{eq_omega_bound_original}
        w_\gamma(t)\le 2e^2~\gamma^2~t~\exp\qty(-\frac{2}{7}\cdot \frac{\gamma t}{\ln^2(\gamma t)})
    \end{align}
     for $\gamma t\ge e^{1/\sqrt{2}}$. 
\end{itemize}
Let  $W_\gamma(t)$ denote the integral of the tail of $w_\gamma(t)$. More precisely, for $t\ge 0$, $W_\gamma(t)$ is defined as,
\begin{align}
    W_\gamma(t)=\int_t^{+\infty} w_\gamma(t')\dd t',
\end{align}
and for $t<0$ we define $W_\gamma(t)=-W_\gamma(-t)$. The $W_\gamma(t)$ function has the following properties:
\begin{itemize}
    \item $W_\gamma(t)$ is continuous and monotone decreasing for $t\ge 0$, with $W_\gamma(0)=1/2$.
    \item there exists constants $t_0>0$ and $c$ such that, for $\gamma t\ge t_0$ one has
    \begin{align}\label{eq_W_bound}
        |W_\gamma(t)|\le c~(\gamma t)^4 \exp\qty(-\frac{2}{7}\cdot \frac{\gamma t}{\ln^2(\gamma t)})
    \end{align}
    \item for $|\omega|>\gamma$, one has $\widetilde{W}_\gamma(\omega)=\frac{-i}{\omega}$.
\end{itemize}

For future reference, it is helpful to separate the almost-exponential part in the upper bounds above and denote that by another symbol. Therefore, we define $u_a(t)$ to denote
\begin{align}
    u_a(t)=e^{-a\frac{t}{\ln^2(t)}}.
\end{align}
In particular, we have the following lemma (Lemma 2.5 in Ref.\cite{bachmann2012automorphic}) about this function,
\begin{lemma}\label{lm_intua}
for all integers $k\ge 0$ and for all $t\ge e^4$ such that also, $a\frac{t}{\ln^2 t}\ge 2k+2$, we have the bound 
\begin{align}
    \int_t^\infty \dd \tau~ \tau^k~ u_a(\tau) \le \frac{(2k+3)}{a}~t^{2k+2}~u_a(t).
\end{align}
\end{lemma}

Throughout this work, we are going to use $\gamma=1/2$. Therefore, it is useful to simplify some of the results for a fixed $\gamma=1/2$. For the rest of this section we denote $W_{1/2}(t)$ by $W(t)$ and $w_{1/2}(t)$ by $w(t)$. In particular, the following statements follow easily from the general bounds stated above:
\begin{itemize}
    \item Based on Eq.\eqref{eq_omega_bound_original}, for $t>2e^{1/\sqrt{2}}>2$, we have $ w(t)\le (e^2/2)\,t~\exp\qty(-\frac{1}{7}\cdot \frac{ t}{(\ln(t)-\ln(2))^2})\le (e^2/2)\,t~\exp\qty(-\frac{1}{7}\cdot \frac{ t}{\ln^2(t)})\le (e^2/2)\,t~\exp\qty(-\frac{1}{7}\cdot \frac{ t}{\ln^2(e^2+t)})$. The reason to add $e^2$ to the argument of the log is to make the exponent well behaved and monotonically decreasing for any $t\ge 0$.  We can also get rid of the $t$ pre-factor by decreasing the exponent. More concretely, clearly there exist a constant $c$ such that for $t>2e^{1/\sqrt{2}}$ we have,
        $w(t)\le c ~\exp\qty(-\frac{1}{14}\cdot \frac{ t}{\ln^2(e^2+t)})$. Finally, since  $\frac{t}{\ln^2(e^2+t)}$ is smooth and bounded for $t\ge 0$, and since $w(t)\le 1$ for $t\ge0$, one can always find a constant $C$, such that,
    \begin{align}
        w(t)\le C \exp\qty(-\frac{1}{14}\cdot \frac{ t}{\ln^2(e^2+t)}),
    \end{align}
    for all $t\ge 0$.
    \item Starting from Eq.\eqref{eq_W_bound}, a similar argument shows that there exists a constant $C'$ such that,
    \begin{align}
        |W(t)|\le C' \exp\qty(-\frac{1}{14}\cdot \frac{ t}{\ln^2(e^2+t)}),
    \end{align}
    for all $t\ge 0$.
    \item similar bounds exists for integrals of $|W(t)|$ as well. To be more concrete, let $G(t)=\int_t^\infty W(t)$ for $t\ge 0$. Then by using the bound in Eq.\eqref{eq_W_bound} and Lemma \ref{lm_intua}, we see that for t larger than some constant $G(t)\le c t^{10} \exp\qty(-\frac{1}{7}\cdot \frac{ t}{\ln^2(t)})$ for a constant $c$. Then based on similar reasoning as above, there exists a constant $C''$ such that for all $t\ge 0$ we have,
    \begin{align}
        \int_t^\infty |W(t)| \dd t \le C'' \exp\qty(-\frac{1}{14}\cdot \frac{ t}{\ln^2(e^2+t)})
    \end{align}
   \end{itemize}

\section{Lieb-Robinson bound on general graphs}\label{apx_LR}
In this section, we state the Lieb-Robinson bound for the time-dependent Hamiltonian evolution on general graphs. We state the theorem for interactions whose tail vanishes as $F(r)$, where $F(r):[0,\infty)\to [0,\infty)$ can be any function that satisfies $F(x)F(y)\le F(x+y)$ and which is monotonic decreasing. Some version of these general bounds are already proven elsewhere (see e.g. Ref.\cite{nachtergaele2019quasi}) we reproduce the proof here because we are proving the bounds under a slightly different assumptions. For the proof, we follow the same strategy as Refs.\cite{hastings2006spectral,bachmann2012automorphic}. One could arrive at better bounds, especially with respect to the
time growth of the bound, by following more involved arguments (see e.g. Refs.\cite{tran2021lieb,chen2021operator,wang2020tightening,nachtergaele2019quasi}), but it is not clear whether those techniques would result in substantially improved bounds on the stability of the gap for the systems of interest.
\begin{proposition}\label{prp_LR_general}
Consider the Hamiltonian $H(t)$,
\begin{equation}
    H(t)=\sum_Z h_Z(t),
\end{equation}
where $Z$ is the subset of sites on which $h_Z(t)$ acts non-trivially. Let $U(t_2,t_1)$ be the time evolution unitary from time $t_1$ to $t_2$:
\begin{align}
    U(t_2,t_1)=\mathcal{T}\qty[\exp(-i \int_{t_1}^{t_2}H(t) \dd t)],
\end{align}
where $\mathcal{T}$ is the time ordering operator. 
Let $F(r):[0,+\infty)\to [0,+\infty)$ be a monotonic decreasing function, such that,
\begin{align}
    F(x)F(y)\le F(x+y).
\end{align}
Let $\norm{h_Z}$ denote $\sup_{s\in [0,t]}\norm{h_Z(s)}$. Assume the following sum converges,
\begin{align}
    \sup_x \sum_{Z\ni x} \frac{\norm{h_Z}|Z|}{F(\text{diam}(Z))}\le s<\infty,
\end{align}
with $s$ a positive constant (independent of time and system size). 
Let $A$ and $B$ be two operators supported on $X$ and $Y$ with $d(X,Y)>0$. Then we have,
\begin{align}
    \norm{[U^\dagger(t,0)\, A\,U(t,0),B]}\le 2\norm{A}\norm{B}|X|\qty(e^{2\,s\,|t|}-1)F(d(X,Y)).
\end{align}
\end{proposition}

\begin{proof}
Let $I_X(t)$ denote terms of the Hamiltonian which act non-trivially on $X$
\begin{align}
    I_X(t)=\sum_{Z:\, Z\cap X\ne \emptyset}h_Z(t).
\end{align}
For simplicity, we define $A(t)$ to denote $U^\dagger(t,0)\, A \, U(t,0)$ and we define $B(t)$ similarly.  We start by finding an expression for $\frac{\\d}{\\dt}\norm{[A(t),B]}$:
\begin{align}
    \norm{[A(t+\epsilon),B]}&=\norm{[[U^\dagger(t,0)U^\dagger(t+\epsilon,t)\, A\, U(t+\epsilon,t)U(t,0),B]}\\
    &=\norm{[[U^\dagger(t+\epsilon,t)\, A\, U(t+\epsilon,t),U(t,0)\,B\,U^\dagger(t,0)]}\\
    &\le \norm{[A+i\epsilon[H(t),A],B(-t)]}+\mathcal{O}(\epsilon^2)\\
    &=\norm{[A+i\epsilon[I_X(t),A],B(-t)]}+\mathcal{O}(\epsilon^2)\\
    &=\norm{[e^{+i\epsilon I_X(t)}Ae^{-i\epsilon I_X(t)},B(-t)]}+\mathcal{O}(\epsilon^2)\\
    &=\norm{[A,e^{-i\epsilon I_X(t)}B(-t)e^{+i\epsilon I_X(t)}]}+\mathcal{O}(\epsilon^2)\\
    &\le \norm{[A,B(-t)-i\epsilon [I_X(t),B(-t)]]}+\mathcal{O}(\epsilon^2)\\
    &\le \norm{[A,B(-t)]}+\epsilon\norm{[A, [I_X(t),B(-t)]]}+\mathcal{O}(\epsilon^2)\\
    &\le \norm{[A(t),B]}+2\epsilon\norm{A}\norm{ [U^\dagger(t,0)\, I_X(t)\, U(t,0),B]}+\mathcal{O}(\epsilon^2).
\end{align}
Therefore, for $t\ge 0$, we find,
\begin{align}
    \norm{[A(t),B]}&\le \norm{[A,B]}+2\norm{A}\int_0^{t} \dd s \norm{ [U^\dagger(s,0)\,I_X(s)\,U(s,0),B]}\label{eq_rec_AB_positive_t}\\
    &\le \norm{[A,B]}+ 2\norm{A}\sum_{Z:\,Z\cap X\ne \emptyset}\int_0^t \dd s \norm{ [U^\dagger(s,0)\,h_Z(s)\,U(s,0),B]}\label{eq_rec_AB}
\end{align}
For a given subset of sites $S$, we define $C_B(S,t)$ as the following,
\begin{align}
    C_B(S,t)=\sup_{O} \frac{\norm{[U^\dagger(t,0)\,O\,U(t,0),B]}}{\norm{O}},
\end{align}
where the supremum is taken over all operators $O$ which only have  support on $S$.  Based on Eq.\eqref{eq_rec_AB} we have,
\begin{align}
    C_B(X,t)&\le C_B(X,0)+2\sum_{Z:\,Z\cap X\ne \emptyset}\int_0^{t} \dd s~\norm{h_Z(s)}~C_B(Z,s)\nonumber \\
    &\le C_B(X,0)+2\sum_{Z:\,Z\cap X\ne \emptyset}\norm{h_Z}\int_0^{t} \dd s~C_B(Z,s)\label{eq_rec_C}
\end{align}
Note that $C_B(Z,0)\le 2\norm{B}$. Moreover, if $Z\cap Y=\emptyset$ then $C_B(Z,0)=0$. Hence, by using \eqref{eq_rec_C} recursively and using the fact that $X\cap Y=\emptyset$ we find,
\begin{align}
    C_B(X,t)\le& 2\sum_{Z_1:\,Z_1\cap X\ne \emptyset}\norm{h_{Z_1}}\int_0^{t} \dd s_1~C_B(Z_1,s_1)\\
    \le&2 \sum_{Z_1:\,Z_1\cap X\ne \emptyset}\norm{h_{Z_1}}~\int_0^{t} \dd s_1~C_B(Z_1,0)\\
    &+2^2\sum_{Z_1:\,Z_1\cap X\ne \emptyset}\norm{h_{Z_1}} \sum_{Z_2:\,Z_2\cap Z_1\ne \emptyset}\norm{h_{Z_2}}\int_0^{t} \dd s_1\int_0^{s_1} \dd s_2~C_B(Z_2,s_2)\\
    \le&2~(2\norm{B})~t~ \sum_{Z_1:\,Z_1\cap X\ne \emptyset, Z_1\cap Y\ne \emptyset}\norm{h_{Z_1}}\\
    &+2^2\sum_{Z_1:\,Z_1\cap X\ne \emptyset}\norm{h_{Z_1}} \sum_{Z_2:\,Z_2\cap Z_1\ne \emptyset}\norm{h_{Z_2}}\int_0^{t} \dd s_1\int_0^{s_1} \dd s_2~C_B(Z_2,s_2)\\
    \le&2~(2\norm{B})~t~ \sum_{Z_1:\,Z_1\cap X\ne \emptyset, Z_1\cap Y\ne \emptyset}\norm{h_{Z_1}}\\
    &+2^2~(2\norm{B})~\frac{t^2}{2!}~\sum_{Z_1:\,Z_1\cap X\ne \emptyset}\norm{h_{Z_1}} \sum_{Z_2:\,Z_2\cap Z_1\ne \emptyset, Z_2\cap Y\ne \emptyset}\norm{h_{Z_2}}\\
    &+2^3~(2\norm{B})~\frac{t^3}{3!}~\sum_{Z_1:\,Z_1\cap X\ne \emptyset}\norm{h_{Z_1}} \sum_{Z_2:\,Z_2\cap Z_1\ne \emptyset}\norm{h_{Z_2}}\sum_{Z_3:\,Z_3\cap Z_2\ne \emptyset, Z_3\cap Y\ne \emptyset}\norm{h_{Z_3}}\\
    &+\cdots
\end{align}
Now we bound each term in the sum. Define $\theta_{Z_1,Z_2}$ to be,
\begin{align*}
    \theta_{Z_1,Z_2}=\begin{cases}
          1, &Z_1\cap Z_2\ne \emptyset\\
          0, &Z_1\cap Z_2= \emptyset
    \end{cases}
\end{align*}
Then we can write the first term in the series as,
\begin{align}
    \sum_{Z_1: Z_1 \cap X\ne \emptyset, Z_1 \cap Y \ne \emptyset }\norm{h_{Z_1}}&\le \sum_{x\in X} \sum_{Z_1 \ni x} \norm{h_{Z_1}}\theta_{Z_1,Y}\\
    &\le \sum_{x\in X} \sum_{Z_1 \ni x} \norm{h_{Z_1}}\frac{F(d(X,Y))}{F(\text{diam}(Z_1))}\theta_{Z_1,Y}\\
    &\le F(d(X,Y)) \sum_{x\in X} \sum_{Z_1 \ni x} \frac{\norm{h_{Z_1}}}{F(\text{diam}(Z_1))}\\
    &\le |X| s_0 F(d(X,Y)),
\end{align}
where $s_0$ is,
\begin{align}
    s_0=\sup_x \sum_{Z_1 \ni x} \frac{\norm{h_{Z_1}}}{F(\text{diam}(Z_1))}
\end{align}
Note that the proposition assumption that $s$ is bounded ensures that $s_0$ would be bounded as well. Similarly, we may bound the double sum as,
\begin{align}
    &\sum_{Z_1: Z_1 \cap X\ne \emptyset}\norm{h_{Z_1}}\sum_{Z_2: Z_2 \cap Z_1\ne \emptyset, Z_2 \cap Y \ne \emptyset }\norm{h_{Z_2}}\le \sum_{x\in X} \sum_{Z_1 \ni x} \sum_{z_1\in Z_1} \sum_{Z_2 \ni z_1}  \norm{h_{Z_1}} \norm{h_{Z_2}} \theta_{Z_2,Y}\\
    &\le \sum_{x\in X} \sum_{Z_1 \ni x} \sum_{z_1\in Z_1} \sum_{Z_2 \ni z_1}  \norm{h_{Z_1}}\frac{F(d(X,z_1))}{F(\text{diam}(Z_1))} \norm{h_{Z_2}}\frac{F(d(z_1,Y))}{F(\text{diam}(Z_2))} \theta_{Z_2,Y}\\
    &\le F(d(X,Y))\sum_{x\in X} \sum_{Z_1 \ni x} \sum_{z_1\in Z_1} \sum_{Z_2 \ni z_1}  \frac{\norm{h_{Z_1}}}{F(\text{diam}(Z_1))} \cdot\frac{\norm{h_{Z_2}}}{F(\text{diam}(Z_2))} \\
    &\le F(d(X,Y))\sum_{x\in X} \sum_{Z_1 \ni x} \sum_{z_1\in Z_1}  \frac{\norm{h_{Z_1}}}{F(\text{diam}(Z_1))} ~s_0 \\
    &\le F(d(X,Y))\sum_{x\in X} \sum_{Z_1 \ni x}  \frac{\norm{h_{Z_1}}~|Z_1|}{F(\text{diam}(Z_1))} ~s_0 \\
    &\le F(d(X,Y))\sum_{x\in X}~s ~s_0 \\
    &\le F(d(X,Y)) |X| ~s~s_0\label{eq_double_bound}\\
\end{align}
To go from the second line to the third line, we used
\begin{align}
    F(d(X,z_1))F(d(z_1,Y))&=F(d(x_0,z_1))F(d(z_1,y_0))\le F(d(x_0,z_1)+d(z_1,y_0))\\
    &\le F(d(x_0,y_0))\le F(d(X,Y)),
\end{align}
where $x_0\in X$ and $y_0\in Y$ are points for which $d(X,z_1)=d(x_0,z_1)$ and $d(z_1,Y)=d(z_1,y_0)$ respectively. The first inequality is obtained by using the assumption about $F(r)$ that $F(x)\,F(y) \le F(x+y)$, the next is obtained from utilizing triangle inequality for distance plus the fact that $F(r)$ is montonic decreasing.  The triple sum can be bounded similarly:
\begin{align}
    &\sum_{Z_1: Z_1 \cap X\ne \emptyset}\norm{h_{Z_1}}\sum_{Z_2: Z_2 \cap Z_1\ne \emptyset}\norm{h_{Z_2}}\sum_{Z_3: Z_3 \cap Z_2\ne \emptyset, Z_3 \cap Y \ne \emptyset }\norm{h_{Z_2}}\\
    &\le \sum_{x\in X} \sum_{Z_1 \ni x} \sum_{z_1\in Z_1} \sum_{Z_2 \ni z_1}\sum_{z_2\in Z_2}\sum_{Z_3\ni z_2}\\
    &\quad\norm{h_{Z_1}}\frac{F(d(X,z_1))}{F(\text{diam}(Z_1))} \norm{h_{Z_2}}\frac{F(d(z_1,z_2))}{F(\text{diam}(Z_2))}
    \norm{h_{Z_3}}\frac{F(d(z_2,Y))}{F(\text{diam}(Z_3))}\theta_{Z_3,Y}\\
    &\le F(d(X,Y)) \sum_{x\in X} \sum_{Z_1 \ni x} \sum_{z_1\in Z_1} \sum_{Z_2 \ni z_1}\sum_{z_2\in Z_2}\sum_{Z_3\ni z_2}\\
    &\quad\frac{\norm{h_{Z_1}}}{F(\text{diam}(Z_1))} \frac{\norm{h_{Z_2}}}{F(\text{diam}(Z_2))}
    \frac{\norm{h_{Z_3}}}{F(\text{diam}(Z_3))}\\
    &\le F(d(X,Y)) |X| s^2 s_0
\end{align}
The rest of the series can be bounded in a similar way. Therefore, we find
\begin{align}
    C_B(X,t)\le 2\norm{B}|X| (s_0/s) \qty(e^{2s\,t}-1)F(d(X,Y))
\end{align}
Lastly, by noting that $s_0/s\le 1$ and using the definition of $C_B(X,t)$, we arrive at,
\begin{equation}
    \norm{[A(t),B]}\le 2\norm{A}\norm{B}|X|\qty(e^{2s\,t}-1)F(d(X,Y))
\end{equation}
For $t<0$, Eq.\eqref{eq_rec_AB_positive_t} will be replaced by,
\begin{align}
    \norm{[A(t),B]}&\le \norm{[A,B]}+2\norm{A}\int_0^{t} (-\dd s) \norm{ [U^\dagger(s,0)\,I_X(s)\,U(s,0),B]}.
\end{align}
The rest of the proof is the same, except that $e^{-2s\,t}$ replaces $e^{-2s\,t}$ in the Lieb-Robinson bound.
\end{proof}

There are two special cases of the above theorem which we use in this work:
\begin{lemma}\label{lm_LR_bound_for_exp_tail}
\textbf{LR bound  for interactions with exponential tails on general graphs:}
Consider the the Hamiltonian $H_s$ to be $H_s=H_0+s V$, where $H_0=\sum_u Q_u$ is a sum of local projectors (acting on nearest neighbors), $s\in[0,1]$ and $V$ has strength $J$ and tail $e^{-\mu r}$. For any two operators $A$ and $B$ with supports $X$ and $Y$ respectively (with $d(X,Y)>0$) we have,
\begin{equation}
     \norm{[e^{iH_s t}Ae^{-iH_st},B]}\le 2\norm{A}\norm{B}|X|e^{\frac{\mu}{4}(v |t|-d(X,Y))},
\end{equation}
with $v$ given by the following sum,
\begin{align}\label{eq_LR_velocity_exp_tail}
    v=\frac{8(e^\mu+s\,J)}{\mu} \sum_r \Gamma(r)^2 e^{-\frac{\mu}{2}r}.
\end{align}
\end{lemma}
\begin{proof}
First note that we can write $H$ as,
\begin{equation}
    H_s=\sum_{u,r} h_{u,r}, \quad \norm{h_{u,r}}\le (e^\mu+s\,J)e^{-\mu r},
\end{equation}
where the $e^\mu$ term in the coefficient appears to account for projections in $H_0$. Then, by using Proposition \ref{prp_LR_general} with $F(r)=\exp(-\frac{\mu}{4}r)$, we find that for operators $A$ and $B$ with supports $X$ and $Y$ respectively (with $d(X,Y)>0$) we have,
\begin{equation}
     \norm{[e^{iH_s t}Ae^{-iH_st},B]}\le 2\norm{A}\norm{B}|X|e^{\frac{\mu}{4}(v |t|-d(X,Y))},
\end{equation}
where $v$ should upper bound the following sum,
\begin{align}
    \frac{8}{\mu} \sup_x \sum_{Z\ni x} \frac{\norm{h_Z}|Z|}{F(\text{diam}(Z))}.
\end{align}
Let $\theta_{Z,x}=\theta_{x,Z}$ be defined to be $1$ if $x\in Z$ and $0$ otherwise. Then,
\begin{align}
    \sum_{Z\ni x} \frac{\norm{h_Z}|Z|}{F(\text{diam}(Z))}&=\sum_{r\ge 0} \sum_{u} \theta_{B_u(r),x} \frac{\norm{h_{u,r}}|B_u(r)|}{F(\text{diam}(B_u(r)))}\\
    &\le  \sum_{r\ge 0} \sum_{u} \theta_{u,B_x(r)} \frac{\norm{h_{u,r}}\Gamma(r)}{F(2r)}\label{eq_lr_velocity_expression}\\
    &\le (e^\mu+s\,J) \sum_{r\ge 0} \Gamma(r)^2 e^{-\mu r/2},
\end{align}
where to get the inequality in the second line, we have used the facts that $\theta_{B_u(r),x}=\theta_{u,B_x(r)}$ and $\text{diam}(B_u(r))\le 2r$.
\end{proof}

\begin{lemma}\label{lm_LR_bound_for_eta_tail}
\textbf{LR bound  for interactions with $\eta$ tails on general graphs:}
Consider the Hamiltonian $V=V(t)$ which for all $t\le 1$ has strength $J$ and tail $c\, r^m\, \Gamma(r)^n \eta_a(r/b)$ for constants $c$, $a>0$, $b>0$, $m$ and $n$. Let $U(0,t)$ be the evolution operator generated by $V$ from $0$ to $t\le 1$. For any operators $A$ and $B$ with supports $X$ and $Y$ respectively (with $d(X,Y)>0$) we have,
\begin{equation}
     \norm{[U(0,t)^\dagger\,A\,U(0,t),B]}\le 2\norm{A}\norm{B}|X|~e^{2 v\,|t|}~\eta_{a/2}(d(X,Y)/(2b)),
\end{equation}
with $v$ given by the following sum,
\begin{align}
    v=c\,J\,\sum_r r^m\Gamma(r)^{2+n}~ \eta_{a/2}(r/b).
\end{align}
\end{lemma}
\begin{proof}
First note that for any $x,y\ge 0$, we have
\begin{align}
    \eta_a(x)\eta_a(y)&=e^{-a~ x/\ln^2(e^2+x)}e^{-a~y/\ln^2(e^2+y)}\\
    &\le e^{-a~x/\ln^2(e^2+x+y)}e^{-a~y/\ln^2(e^2+y+x)}=e^{-a (x+y)/\ln^2(e^2+x+y)}\\
    &=\eta_a(x+y).
\end{align}
Clearly the same is true for $\eta_a(x/b)$ for $b>0$. Also, it is easy to see that $\frac{x}{\ln^2(e^2+x)}$ is monotonic increasing for $x\ge 0$,
\begin{align}
    \frac{\dd}{\dd x}\qty[\frac{x}{\ln^2(e^2+x)}]=\frac{\ln(e^2+x)-2+\frac{2e^2}{e^2+x}}{\ln^3(e^2+x)}>0,
\end{align}
and thus $\eta_a(r/b)$ is monotonic decreasing, for $b> 0$. Therefore, by using Proposition \ref{prp_LR_general} with $F(r)=\eta_{a/2}(r/(2b))$, we find that for operators $A$ and $B$ with supports $X$ and $Y$ respectively (with $d(X,Y)>0$) we have,
\begin{equation}
     \norm{[e^{iH_s t}Ae^{-iH_st},B]}\le 2\norm{A}\norm{B}|X|~e^{2 v\,|t|}~\eta_{a/2}(d(X,Y)/(2b)),
\end{equation}
with $v$ given by the following sum,
\begin{align}
    v=cJ \sum_r r^m~\Gamma(r)^{2+n}~ \eta_{a/2}(r/b),
\end{align}
which readily follows from   substituting $F(r)=\eta_{a/2}(r/2b)$ in Eq.\eqref{eq_lr_velocity_expression}.
\end{proof}

\subsection{Lieb-Robinson bounds and local decomposition}
Let $A$ be an operator almost supported on $X=B_{u_0}(r_0)$, meaning that for any operator $O$ with $d_O=d(X,\text{supp}(O))>0$ one has the following bound
\begin{align}
 \norm{[A,O]}\le c_A\norm{O}|X| f(d_O),
\end{align}
with $f(r)$ is a decaying function of $r$ and $c_A$ is a constant independent of $O$. As shown in \cite{bravyi2006lieb}, if $f(r)$ is decaying fast enough, one can use this bound to find a local decomposition for $A$
\begin{align}
    A=\sum_{r\ge r_0} A_{u_0,r},
\end{align}
where $A_{u,r}$ is supported on $B_u(r)$ and its norm vanishes as one increases $r$. To see this, define $A'_{u_0,r} $ for $r\ge r_0$ to be 
\begin{align}
    A'_{u_0,r}=\frac{1}{\tr(I_{S(r)})}\tr_{S(r)}(A)\otimes I_{S(r)}
\end{align}
where $S(r)$ is the complement of $B_{u_0}(r)$ and $I_{S(r)}$ is the identity operator acting on the qubits in region $S(r)$.
 $A'_{u_0,r}$ is clearly supported only on $B_{u_0}(r)$. Equivalently, $A'_{u_0,r}$ can be written as,
\begin{align}
    A'_{u_0,r}=\int \dd \mu(U) ~ U\, A U^\dagger , 
\end{align}
where the integral is over all unitaries that act trivially on $B_{u_0}(r)$ and $\mu(U)$ is the Haar measure. Therefore we have
\begin{align}
    \norm{A-A'_{u_0,r}}\le \int \dd \mu(U) \norm{[A,U]} \le c_A \Gamma(r_0) f(r-r_0+1),
\end{align}
where we have used the fact that $d(X,S(r))\ge r-r_0+1$, and that $f$ is a monotone decreasing function. This in turn shows that we can approximate $A$ by a strictly local operator $A'_{u_0,r}$ with an error that is controlled by $f$. To get an exact local decomposition, we can simply define $A_{u_0,r}$ as,
\begin{align}\label{eq_O_ur}
    A_{u_0,r}=A'_{u_0,r}-A'_{u_0,r-1},
\end{align}
for $r\ge r_0+1$, and define $A_{u_0,r_0}=A'_{u_0,r_0}$. Then it is clear that,
\begin{align}
    A=\sum_{r\ge r_0} A_{u_0,r},
\end{align}
where for $r>r_0$ we have,
\begin{align}
    \norm{A_{u_0,r}}&= \norm{A'_{u_0,r}-A'_{u_0,r-1}}\le \norm{A-A'_{u_0,r}}+\norm{A-A'_{u_0,r-1}}\\
    &\le 2 c_A \Gamma(r_0)f(r-r_0).
\end{align}
We can use this result to find the new local decomposition of a quasi-local operator under unitary evolution. We state that as a Lemma for future references:

\begin{lemma}\label{lm_LR_local_decom}
\textbf{Locality of operators under transformations with a LR bound:}
Suppose $\mathcal{L}$ is a linear superoperator such that for any operator $A$,
\begin{align}
    \norm{\mathcal{L}(A)}\le c\norm{A},
\end{align}
where $c$ is a constant.  Moreover, assume that if $A$ is supported on $B_u(r)$, then $\mathcal{L}(A)$ is almost supported on $B_u(r)$ as well, meaning that for any operator $O$ with $d_O=d(B_{u}(r),\text{supp}(O))>0$ the following LR bound holds,
\begin{align}
 \norm{[\mathcal{L}(A),O]}\le C \norm{A}\norm{O}\Gamma(r) g(d_O),
\end{align}
where $g$ is a decreasing function and $C$ is a constant. Let $V_u$ be an operator localized around $u$ with strength $J$ and tail $f(r)$. Then $\mathcal{L}(V_u)$ is localized around $u$ and has strength $J$ and tail $h(r)$, where
\begin{align}
    h(r)=\Bigg(c+ C~\Big[f(1)+g(1)\Big]\Bigg)~r~\Gamma(r)\max\{f(r/2),g(r/2)\}.
\end{align}
In particular, if $V$ is an operator with strength $J$ and tail $f(r)$, $\mathcal{L}(V)$ has strength $J$ and tail $h(r)$.
\end{lemma}

\begin{proof}
Since $V$ has strength $J$ and tail $f(r)$ we can write $V$ as $V=\sum_{u,r} V_{u,r}$ with $V_{u,r}$ supported on $B_u(r)$ and $|V_{u,r}|\le J~f(r)$.
Let $V'_{u,r;\rho}=\frac{1}{\tr(I_{S(r)})}\tr_{S(r)}(\mathcal{L}( V_{u,\rho}))\otimes I_{S(r)}$, for $r\ge \rho$ with $S(r)$ denoting the complement of $B_u(r)$. Let $V''_{u,r;\rho}=V'_{u,r;\rho}-V'_{u,r;\rho-1}$ for $r\ge \rho+1$ and let $V''_{u,\rho;\rho}=V'_{u,\rho;\rho}$. The discussion at the beginning of this section shows that,
\begin{align}
   \mathcal{L}(V_{u,\rho}) = \sum_{r\ge \rho} V''_{u,r;\rho},
\end{align}
where $V''_{u,r;\rho}$ is supported on $B_u(r)$ and $\norm{V''_{u,r;\rho}}\le 2C J f(\rho)\Gamma(\rho)g(r-\rho)$ for $r>\rho$. For $r=\rho$ we have,
\begin{align}
    \norm{V''_{u,\rho;\rho}}=\norm{V'_{u,\rho;\rho}}=\norm{\int \dd \mu(U) U \mathcal{L}(V_{u,\rho}) U^\dagger}\le \norm{\mathcal{L}(V_{u,\rho})}\le c J f(\rho).
\end{align}
We define $ \widetilde{V}_{u,r}$ as
\begin{align}
    \widetilde{V}_{u,r}=\sum_{\rho\le r}V''_{u,r;\rho}.
\end{align}
Note that
\begin{align}
    \mathcal{L}(V_u)&=\sum_\rho  \mathcal{L}(V_{u,\rho}) = \sum_\rho \sum_{r\ge \rho} V''_{u,r;\rho}\\
    &=\sum_r \qty (\sum_{\rho\le r}V''_{u,r;\rho})=\sum_{r} \widetilde{V}_{u,r}.
\end{align}
Moreover, 
\begin{align}
    \norm{\widetilde{V}_{u,r}}&\le\sum_{\rho\le r} \norm{V''_{u,r;\rho}}\le J\,\qty[cf(r)+ 2C \sum_{\rho<r}  f(\rho)\Gamma(\rho)g(r-\rho)]\\
    &\le J\,\Bigg[cf(r)+ 2C \sum_{\rho< r/2}  f(\rho)\Gamma(\rho)g(r-\rho)+2C \sum_{r/2\le \rho< r}  f(\rho)\Gamma(\rho)g(r-\rho)\Bigg]\\
    &\le J\,\Bigg[cf(r/2)+ 2C~\frac{r}{2}\cdot f(1)\Gamma(r/2)g(r/2)+2C~\frac{r}{2}\cdot f(r/2)\Gamma(r)g(1)\Bigg]\\
    &\le J\,\Bigg[c+ C~\Big[f(1)+g(1)\Big]\Bigg]r\Gamma(r)\max\{f(r/2),g(r/2)\}
\end{align}
\end{proof}

\subsection{Lieb-Robinson bounds for filtered interactions}
\begin{lemma}\label{lm_LR_bound_for_filters}
Let $A$ be a local operator supported on $X$. Let $U(t)$ be a unitary operator such that for any operator $O$ with $d_O=(X,\text{supp}(O))>0$ one has the following LR bound,
\begin{align}
    \norm{[U(t)^\dagger A U(t), O]}\le c ~\norm{A}~ \norm{O} ~|X|~ e^{\mu(v |t|-d_O)},
\end{align}
for some positive constants $c$,$\mu$ and $v$. Assume $F(t)$ is a function such that $\int \dd t' |F(t')|$ is finite and for any $t>0$, one has $\int_{|t'|>t} \dd t' |F(t')|\le c'~\eta_{a}(t)$ with $c'$ and $a$  positive constants. Let $\widetilde{A}$ denote the following
\begin{align}
    \widetilde{A}=\int_{-\infty}^{+\infty} \dd t F(t) U^\dagger(t) A U(t)
\end{align}
Then, for any operator $O$ with $d_O=(X,\text{supp}(O))>0$ we have the following LR bound for $\widetilde{A}$,
\begin{align}
    \norm{[\widetilde{A}, O]}\le C~\norm{A}~\norm{O}~|X|~\eta_{a}(d_O/(2v)),
\end{align}
for some constant $C$ which depends on $c$, $c'$, $\mu$, $v$ and $a$.
\end{lemma}

\begin{proof}
To see that, we can simply divide the integral into two parts:
\begin{align}
    \norm{[\widetilde{A},O]}&\le \int \dd t |F(t)|~\norm{[U^\dagger(t)~A~U(t),O]}\\
    &=\int_{|t|<d_O/(2v)} \dd t \cdots +\int_{|t|>d_O/(2v)} \dd t \cdots\\
    &\le c ~\norm{A}~ \norm{O} ~|X|~ e^{-\frac{\mu}{2}d_O} \int \dd t |F(t)|+2\norm {A}\norm{O}\int_{|t|>d_O/(2v)}\dd t~|F(t)|\\
    &\le  \norm {A}\norm{O}\qty[ c''\,|X|\, e^{-\frac{\mu}{2} d_O} + 2\,c'\,\eta_{a}(d_O/(2v))]\\
    &\le  \norm {A}\norm{O}|X|\qty[ c''\, e^{-\frac{\mu}{2} d_O} + 2\,c'\,\eta_{a}(d_O/(2v))]\\
    &\le C \norm {A}\norm{O}|X|\eta_{a}(d_O/(2v))
\end{align}
where in the last part we have used the fact that one can always find a constant $C'$ (which could depends on $\mu$, $a$ and $v$), such that $ e^{-\frac{\mu}{2} r}\le C'~\eta_{a}(r/(2v))$ for all $r\ge 0$.
\end{proof}

\subsection{Proof of Lemma \ref{lm_loc_decom_filter}}
\begin{proof}
Let $A$ and $B$ be local operators supported on $X$ and $Y$ respectively, with $d(X,Y)>0$. Due to Lemma \ref{lm_LR_bound_for_exp_tail}, we have the following LR bound,
\begin{align}
    \norm{[e^{i\,H\,t}Ae^{-i\,H\,t},B]}\le 2\norm{A}\norm{B}|X|e^{\frac{\mu}{4}(v |t|-d(X,Y))},
\end{align}
with $v$ given as $v=\frac{8 J_1}{\mu}\sum_{r\ge 1}\Gamma(r)^2 e^{-\frac{\mu}{2}r}.$ Define $\mathcal{L}(A)$ as,
\begin{align}
    \mathcal{L}(A)=\int_{-\infty}^{\infty}\dd t~ F(t)~ e^{i\,H\,t}~A~e^{-i\,H\,t}.
\end{align}
Since $F(t)\in L^1(\mathbb{R})$, we have
\begin{align}
    \norm{\mathcal{L}(A)}\le \qty(\int \dd t |F(t)| )\norm{A}\le c_2 \norm{A},
\end{align}
for a constant $c_2$. Based on Lemma \ref{lm_LR_bound_for_filters}, $\mathcal{L}(A)$ satisfies the following LR bound,
\begin{align}
    \norm{[\mathcal{L}(A), B]}\le c_3~\norm{A}~\norm{B}~|X|~\eta_{a}(d(X,Y)/(2v)),
\end{align}
for a constant $c_3$ which depends on $\mu$,$v$ and $a$.  Therefore, based on Lemma \ref{lm_LR_local_decom}, $\widetilde{V}_u$ is localized around $u$, has strength $J$ and tail $\tilde{f}(r)$ given as,
\begin{align}
    \tilde{f}(r)=c\, r\, \Gamma(r) \max\{f(r/2),\eta_{a}(r/(4v))\},
\end{align}
for a constant $c$ which depends on $\mu$,$v$ and $a$.
\end{proof}

\subsection{Proof of Lemma \ref{lm_loc_decom_eta_unitary}}
\begin{proof}
Let $A$ and $B$ be local operators supported on $X$ and $Y$ respectively, with $d(X,Y)>0$. Then, based on Lemma \ref{lm_LR_bound_for_eta_tail}, $U^\dagger_t\, A\, U_t$ satisfies the following LR bound,
\begin{equation}
     \norm{[U_t^\dagger\,A\,U_t,B]}\le 2\norm{A}\norm{B}|X|~e^{2 v\,|t|}~\eta_{a/2}(d(X,Y)/(2b)),
\end{equation}
with $v$ given as $v=c_1\,J\,\sum_r r\Gamma(r)^{3}~ \eta_{a/2}(r/b)$.  Also trivially we have $\norm{U^\dagger_t\,A\,U_t}=\norm{A}$ for any operator $A$. Therefore, based on Lemma \ref{lm_LR_local_decom}, $U_t^\dagger V_u U_t$ is localized around $u$, has strength $J$ and tail $\tilde{f}(r)$ given as,
\begin{align}
    \tilde{f}(r)=c\, r\, \Gamma(r)\max\{f(r/2),\eta_{a/2}(r/(4b))\},
\end{align}
for a constant $c$ which depends on $a$, $b$, $c_1$ and $v t$.
\end{proof}

\section{Proof of Lemma \ref{lm_lcgc}}\label{apx_proof_of_lm_lcgc}
\begin{proof}
Let $\ket{\phi}$ and $\ket{\psi}$ be the states such that $\norm{O P}=\norm{O P\ket{\phi}}$ and $\norm{O P_C}=\norm{O P_C \ket{\psi}}$. Clearly, we should have $P\ket{\phi}=\ket{\phi}$ and $P_C\ket{\psi}=\ket{\psi}$. But it means that  $P_B\ket{\phi}=\ket{\phi}$ and  $P_B\ket{\psi}=\ket{\psi}$. Then, we have
\begin{align}
    \norm{O P}&=\tr(O^\dagger O \ketbra{\phi})^{1/2}=\tr(P_B\, O^\dagger O P_B \ketbra{\phi})^{1/2}\\
    &=\tr(cP_B\ketbra{\phi})^{1/2}=c^{1/2}=\tr(cP_B\ketbra{\psi})^{1/2}\\
    &=\tr(P_B\,O^\dagger O P_B\, \ketbra{\psi})^{1/2}=\tr(\,O^\dagger O\, \ketbra{\psi})^{1/2}=\norm{O P_C},
\end{align}
\end{proof}

\end{document}